\pdfoutput=1
\documentclass[nonacm,acmsmall,final,screen,natbib=false]{acmart}

\usepackage{amsmath, amsthm, amssymb, calc}
\usepackage{mathtools} %
\usepackage{thm-restate}
\usepackage{url}
\usepackage{enumitem}
\usepackage{framed} 
\usepackage[small, bold, full]{complexity} \newclass{\EXPTIME}{EXPTIME}

\usepackage{bm}

\usepackage{dsfont}

\usepackage{multirow}
\usepackage{rotating}

\usepackage{ifdraft}
\usepackage[colorinlistoftodos,prependcaption]{todonotes}

\usepackage{hyperref}

\usepackage[capitalise,nameinlink]{cleveref}

\renewcommand{\S}{\textsection}
\renewcommand{\top}{\ensuremath\mathsf{T}}
\renewcommand{\epsilon}{\ensuremath\varepsilon}

\newcommand{\Id}{\operatorname{Id}}

\newcommand{\Z}{\mathbb{Z}}

\newcommand{\N}{\mathbb{N}}

\newcommand{\Q}{\mathbb{Q}}
 
\newcommand{\F}{\mathbb{F}}

\newcommand{\NN}{\mathbb{N}}
\newcommand{\ZZ}{\mathbb{Z}}
\newcommand{\QQ}{\mathbb{Q}}

\newcommand{\QQbar}{\smash{\overline{\Q}}\vphantom{\Q}}

\newcommand{\FFbar}{\smash{\overline{\F}}\vphantom{\F}}

\newcommand{\diag}[1]{\ensuremath{\Delta(#1)}}
\DeclareMathOperator{\GL}{GL}

\providecommand*{\eu}%
{\ensuremath{\mathrm{e}}}
\providecommand*{\iu}%
{\ensuremath{\mathrm{i}}}

\newcommand{\alg}{\smash{\overline{\Q}}}

\newcommand{\appref}[1]{\hyperref[#1]{Appendix~\ref*{#1}}}
\usepackage[plain]{algorithm}
\usepackage{algorithmicx,algpseudocode}

\usepackage{tabularx}
\usepackage{color, colortbl}
\definecolor{Gray}{gray}{0.9}
\definecolor{LightCyan}{rgb}{0.88,1,1}

\usepackage{libertine}

\RequirePackage[
backend=biber,
datamodel=acmdatamodel,
url=false,
doi=true,
giveninits=true,
style=acmnumeric,
sortcites
]{biblatex}

\DeclareSourcemap{
	\maps[datatype=bibtex]{
		\map{
			\step[fieldset=issn, null]
		}
	}
}

\AtEveryBibitem{
	\clearlist{address}
	\clearfield{isbn}
	\clearfield{issn}
	\clearlist{location}
	\clearfield{month}
	\clearfield{series}
	
}

\AtEndPreamble{%
	\theoremstyle{acmplain}
    
	\newtheorem{claim}[conjecture]{Claim}
	
	\theoremstyle{acmdefinition}
    \newtheorem{remark}[conjecture]{Remark}
}

\begin{document}

 	\author{Rida Ait El Manssour}
 	\orcid{0000-0001-6228-9071}
 	\affiliation{%
 		\institution{University of Oxford}
 		\city{Oxford}
 		\country{United Kingdom}
 	}
 	\email{rida.aitelmanssour@cs.ox.ac.uk}
	
 	\author{George Kenison}
 	\orcid{0000-0002-7661-7061}
 	\affiliation{%
 		\institution{Liverpool John Moores University}
 		\city{Liverpool}
 		\country{United Kingdom}
 	}
 	\email{g.j.kenison@ljmu.ac.uk}
	
 	\author{Mahsa Shirmohammadi}
 	\orcid{0000-0002-7779-2339}
 	\affiliation{%
 		\institution{CNRS - IRIF}
 		\city{Paris}
 		\country{France}
 	}
 	\email{mahsa@irif.fr}
	
 	\author{Anton Varonka}
 	\orcid{0000-0001-5758-0657}
 	\affiliation{%
 		\institution{TU Wien}
 		\city{Vienna}
 		\country{Austria}
 	}
 	\email{anton.varonka@tuwien.ac.at}
 	
 	 	\author{James Worrell}
 	\orcid{0000-0001-8151-2443}
 	\affiliation{%
 		\institution{University of Oxford}
 		\city{Oxford}
 		\country{United Kingdom}
 	}
 	\email{jbw@cs.ox.ac.uk}

	

\settopmatter{printfolios=true}

\title{Determination Problems for Orbit Closures and Matrix Groups}

\begin{abstract}
Computational problems concerning the orbit of a point
under the action of a matrix group occur throughout
computer science, including in program analysis, complexity theory, quantum computation, and automata theory.
In many cases  the focus extends beyond orbits proper to orbit closures under a suitable topology.
Typically one starts from a group and a set of points and asks questions
about the orbit closure of the set under the action of the group, e.g., whether two given orbit closures intersect.

In this paper we consider a collection of what we call determination problems concerning matrix groups and orbit closures.  These problems begin with  a given variety and 
seek to understand
whether and how it 
arises either as an algebraic matrix group or as an orbit closure. 
The \emph{how} question asks whether the underlying group is 
$s$-generated, meaning it is topologically generated by~$s$ matrices for a given number~$s$. 
Among other applications,  problems of this type have recently been studied in the context of synthesising loops subject to certain specified invariants on program variables.

Our main result is a polynomial-space procedure
that inputs a variety  and a number~$s$ and determines whether the given  variety arises as an orbit closure of a point under 
an $s$-generated commutative algebraic matrix group. 
The main tools in our approach are structural properties of commutative algebraic matrix groups and module theory. We leave open the question of determining whether a variety is an 
orbit closure of a point under an $s$-generated algebraic matrix group (without the requirement of commutativity).

\end{abstract}

\maketitle

\section{Introduction}
\label{sec:intro}

\subsection{Orbits and their Closures}
The general linear group $\GL_d(\mathbb{F})$ is the group of all $d\times d$ invertible matrices with entries in a field~$\mathbb{F}$.
 Let $G$ be a subgroup  of~$\GL_d(\mathbb{F})$, acting on $\mathbb{F}^d$ via left multiplication, that is,~$g\in G$ maps~$\bm{v} \in \mathbb{F}^d$ to $g\cdot \bm{v}$.
 The \emph{orbit} of~$\bm{v} \in \mathbb{F}^d$ under~$G$ is the set $G \cdot \bm{v}\coloneqq \{g \cdot \bm{v} \, |\, g \in G\}$.
The computational study of orbits of  subgroups of $\GL_d(\mathbb{F})$ stretches back many decades.
One of the most fundamental problems is determining whether a given pair of vectors~$\bm{u},\bm{v} \in \mathbb{F}^d$ lie
in the same orbit under the action of a finitely generated subgroup~$G$ of~$\GL_d(\mathbb{F})$, 
that is, whether some element $g \in G$ maps $\bm{v}$ to $\bm{u}$.
Over the field $\mathbb Q$ this problem is undecidable in general, but it is
decidable if~$G$ is commutative, and even in polynomial time if $G$ is generated by a single matrix~\cite{BBCIL96,kannan1986orbit}.

For many applications it makes sense to study \emph{orbit closures} in place of orbits.

Orbit closures can be seen as an instance of abstract semantics in program analysis, where one 
over-approximates the set of reachable program states by a set from a particular abstract domain (such as intervals, octagons, or polyhedra).
In the case at hand, the abstract domain is the collection of \emph{algebraic subsets} of {$\FFbar^d$ (where $\FFbar$ is the algebraic closure of $\mathbb F$)}, where
{$S\subseteq \FFbar^d$} is algebraic if it arises as the set of common zeros of a collection of polynomials in
$\mathbb F[x_1,\ldots,x_d]$.  One can thus think of algebraic sets as surfaces in  {$\FFbar^d$}.

The algebraic subsets of {$ \FFbar^d$} form the closed sets of a topology, called the \emph{Zariski topology}.
Formally,
the \emph{orbit closure} of $\bm{v} \in \mathbb F^d$ under the action of a subgroup $G$ of~$\GL_d(\mathbb{F})$
is the closure $\overline{G \cdot \bm{v}}$ of the orbit 
with respect to the Zariski topology, that is, $\overline{G \cdot \bm{v}}$ is the smallest algebraic superset of the orbit $G \cdot \bm{v}$.  
When $\mathbb F$ is the field of complex numbers, the 
Zariski topology is coarser than the familiar Euclidean topology, and thus the Zariski closure contains the Euclidean closure.  However   
the Zariski and Euclidean closures of an orbit coincide when $G$ is a
linear algebraic group (that is, when $G$ is both a group and algebraic set, see~\cref{sec:preliminaries}).

The orbit closure of a point $\bm{w}\in\mathbb F^d$ is determined by the set of polynomial relationships that hold on every point 
in the orbit of $\bm{w}$.
  For instance, consider the following matrix~$M$ and the vector $\bm{w}$:
 \[M:=\begin{pmatrix}
     1& 1\\
   1&0
 \end{pmatrix} \qquad \qquad \bm{w}=\begin{pmatrix}
    1\\
    0
\end{pmatrix}.\]
The orbit of $\bm{w}$ under the group generated by $M$ consists of the vectors~$M^n\cdot \bm{w}=(
    F_{n+1},
    F_n
)^\top$, where $F_n$ is the $n$th Fibonacci number for $n\in \mathbb Z$. 
Based on the identity $(F_{n+1}^2-F_{n+1}F_n-F_n^2)^2=1$, which holds for all $n\in \mathbb Z$,
the orbit closure is the curve 
$\{(x,y):(x^2-xy-y^2)^2=1\}$.

Orbit closure problems arise across many areas of computer science, including complexity theory, program analysis, quantum computation and automata theory~\cite{DerksenJK05, hrushovski2023strongest,SeidlMK15,HrushovskiOP018, KincaidCBR18,CyphertK24,SankaranarayananSM04,Muller-OlmS04,BurgisserLMW11,burgisser2024completeness}. 
Among many computational questions studied in this context are orbit-closure containment and intersection: the former  asks whether a vector $\bm{u} \in \mathbb{F}^d$ is contained in the orbit closure \(\overline{G \cdot \bm{v}}\) of another vector~$\bm{v}$, and the latter asks whether two such closures intersect.
A striking application arises in geometric complexity theory, where the \(\VP=\VNP\) problem\footnote{The {\bf VP}$=${\bf VNP} problem  is the algebraic variant of {\bf P}$=${\bf NP} problem.} has been studied in terms of orbit-closure containment with respect to the
action of general linear group~$\GL_d(\QQbar)$ on polynomial rings~\cite{BurgisserLMW11,burgisser2024completeness}.

In certain applications, such as non-convex optimisation problems,  non-commutative
rational identity testing, and  graph isomorphism~\cite{forbes2013conjug, derksen2020algorithms, blaser2021orbit, burgisser2021polynomial}, one considers the orbit closure of a vector~$\bm{v}$ under the action of an explicitly given  algebraic group~$G$,   which is  specified as the set of common zeros of a finite collection of polynomials.
In other
contexts, such as quantum computing and program analysis~\cite{DerksenJK05,hrushovski2023strongest,aitelmanssour2025loops}, the goal is to compute the orbit closure of a group $G$ that is \emph{implicitly} presented via a
finite set of topological generators.
More precisely, the group is expressed as   $G=\overline{\langle M_1,\ldots,M_s \rangle}$ with the matrices~$M_i$ given as  input.
We  refer to these two settings as \emph{explicit} and \emph{implicit} presentations of the orbit-closure problems, respectively.

 The explicit orbit-closure containment  and intersection problems\footnote{In~\cite{blaser2021orbit}, the problems are simply called orbit-closure containment and intersection.} can  be directly formulated as existential formulas of first-order logic over the base field~$\mathbb{F}$~\cite{blaser2021orbit}.
 In the other direction, it was
shown in~\cite{blaser2021orbit}  that orbit containment over $\mathbb{R}$  is polynomial-time equivalent to the existential theory of the reals, and it is \(\NP\)-hard over~$\QQbar$.
Further applications of explicit orbit-closure problems have been identified in~\cite{burgisser2021polynomial}, 
particularly in combinatorial optimisation and dynamical systems, under the assumption that the underlying group is commutative. 
The complexity of these problems is fully resolved in the case of a key subclass of commutative groups (namely, tori) but remains open for general commutative groups~\cite{burgisser2021polynomial}.

In the implicit orbit-closure problem, the main challenge in computing an orbit closure lies in computing a set of polynomials whose zero set is the Zariski closure of the group in question. 
For matrices over $\QQ$
an algorithm for this task  was given in~\cite{DerksenJK05}. This algorithm is not known to be elementary~\cite[Appendix C]{NPSHW2021};  
an elementary procedure albeit of severalfold exponential time was recently  provided in~\cite{NPSHW2021}.
Recent progress on computing multiplicative relations among algebraic numbers~\cite{combot2025computing} shows that for the special case of cyclic groups, the approach in~\cite[Theorem 1.1]{aitelmanssour2025loops}, which originally results in a polynomial-space algorithm, gives
a polynomial-time upper bound for the implicit orbit-closure computation of cyclic groups (and cyclic semigroups).  
 It remains a challenging open problem to close the complexity gap for the explicit orbit-closure problems in the general setting.

The study of 
implicit orbit-closure problems in 
quantum computation and automata theory have led to the resolution 
of the  equivalence problem for deterministic top-down  
tree-to-string transducers~\cite{SeidlMK15} and 
the threshold problem for quantum automata~\cite{DerksenJK05}. 
Orbit closures feature prominently in  program analysis when one
wants to automatically compute polynomial invariants of certain classes of loop programs~\cite{HrushovskiOP018, KincaidCBR18,CyphertK24,SankaranarayananSM04,Muller-OlmS04}.

\subsection{Determination Problems}

In this paper, we investigate a series of determination problems related to groups and their orbit closures. These problems start with a given algebraic set (commonly referred to as a variety) and examine whether it can be realized as a linear algebraic group or as an orbit closure, with the constraint that the underlying group be topologically $s$-generated. 
Here we define an algebraic group $G \leq \GL_d(\QQbar)$ to be 
\emph{topologically $s$-generated} if {there is a set $S \subseteq \GL_d(\QQbar)$  of cardinality  $s$ } such that  
 $G=\overline{\langle S\rangle}$.
By \cref{prop:gen},
algebraic groups are always  topologically generated by a finite set.

In this context, 
determining whether a variety $Z\subseteq \QQbar^d$ arises as an orbit closure under the action of $G$ is, in principle, straightforward. 
Let $\mathrm{Sym}(Z)$ denote the  subgroup of $\GL_d(\QQbar)$ consisting of all matrices that leave~$Z$ invariant, defined as $\mathrm{Sym}(Z)\coloneqq \{ A \in
\GL_d(\QQbar) : A(Z)=Z\}$.
Any group~$G$ having $Z$ as an orbit-closure 
 is necessarily a subgroup of $\mathrm{Sym}(Z)$,
and so we
 may assume without loss of generality that $G$ is $\mathrm{Sym}(Z)$.
But $\mathrm{Sym}(Z)$ is
definable in first-order logic over~$\QQbar$ and hence the orbit closure of any 
point of $Z$ under $\mathrm{Sym}(Z)$ is definable over the real closed field $\QQbar\cap\mathbb R$.\footnote{The Zariski closure of every first-order definable 
set over~$\QQbar$ is equal to its Euclidean closure and is thus
first-order definable over $\QQbar\cap\mathbb R$ after identifying 
each element of $\QQbar$ as the pair of its real and imaginary parts.}
Hence the question of whether $Z$ arises as the orbit closure of a point under
$\mathrm{Sym}(Z)$ reduces to the decision problem for the theory of real closed fields.
It further holds by \cref{prop:gen} that 
 $\mathrm{Sym}(Z)$ has a finitely generated subgroup whose Zariski closure is~$\mathrm{Sym}(Z)$ itself.  Whence $Z$ is the orbit closure
of a point under some finitely generated matrix group if
and only if it is the orbit closure of a point under
$\mathrm{Sym}(Z)$.
 However, it is more challenging to determine whether a given variety is the orbit closure of a topologically $s$-generated group than simply determining whether it is an orbit closure.

Our main  determination problems are  as follows:
\begin{itemize}
    \item The \textbf{Group Determination} problem asks,  given 
  $s \in \N$ and a family
 of $m$ polynomials in $\QQ[\{x_{i,j}\}_{1\leq i,j\leq d}]$, each
 of total degree at most $b$,
 to determine whether their zero locus
 $Z\subseteq\GL_d(\QQbar)$ is an $s$-generated algebraic matrix group.
 \item The \textbf{Orbit-closure Determination} problem asks,  given 
  $s \in \N$ and a family
 of $m$ polynomials in $\QQ[\{x_i\}_{1\leq i\leq d}]$, each
 of total degree at most $b$,
  determine whether 
 their zero locus~$Z\subseteq \QQbar^d$ is 
 the orbit closure of some point~$\bm{v}\in \QQbar^d$ under the action of an 
 $s$-generated algebraic matrix group. 
\end{itemize}
In our complexity analysis we refer to  the tuple $(s,d,m,b)$ as the parameters of the problem instances. By~\cref{prop:cyclic,prop:uni,Theremark},  the minimum number~$s$ of topological generators for the groups we study, commutative algebraic groups, is upper  
bounded by~$d$.

This paper focuses on addressing the complexity of the determination problems for commutative algebraic matrix groups. 
Our main result, \cref{pro:loopGen}, is a polynomial-space procedure for both problems, based on reductions to the decision problem for formulas in a fragment of the first-order theory of algebraically closed or real closed fields of characteristic zero.

\begin{example} \label{ex:loopGenintro}
Let $Z\subseteq \QQbar^4$ be the set of common zeros of the two polynomials
\(Q_1:=x_2^2 - x_1 - x_4\) and \(Q_2:= -2 x_4 x_2 - 2x_3^2 - \frac{1}{5} x_2 x_3\).
The task underlying this instance of the orbit-closure determination problem (as above) is to determine whether \(Z\) is the orbit closure of some point \(v\in\QQbar^4\) under the action of an algebraic matrix group with a prescribed number of generators.

Our nondeterministic procedure in \cref{pro:loopGen} shows that \(Z\) is the orbit closure of a 1-generated algebraic matrix group, i.e., \(Z=\overline{\langle M\rangle \cdot \bm{v}}\).
The following matrix \(M\) and vector \(\bm{v}\) pair witness that $Z$ is the orbit-closure of a one-generated (cyclic) matrix group: 
\[
M =  \begin{pmatrix}
    25& \phantom{-}0& -1& 20\\
    0& \phantom{-}5& \phantom{-}0& 0 \\
    0& -\frac{1}{2}& \phantom{-}5& 0\\
    0& \phantom{-}0& \phantom{-}1& 5
\end{pmatrix} \quad \text{ and } \quad \boldsymbol{v}=  \begin{pmatrix}
    1\\
    1\\
    0\\
    0
\end{pmatrix}.
\]
An account of the steps taken to produce $M$ and $\bm{v}$ is given in~\cref{ex:loopGen}. \hfill $\blacktriangleleft$
\end{example}

The extension of our results from the case of commuting matrices to 
 the case of general algebraic matrix groups appears to be  challenging.
To approach the above version of the orbit determination problem, we rely on the observation that with respect to a convenient basis an orbit closure itself carries the structure of an algebraic matrix group.  We then  
examine structural properties of
 semisimple and unipotent linear algebraic groups to identify  when the orbit, seen as a group, arises as the closure of a commutative group.

\subsection{Application in Programming Languages}

A research direction closely related to orbit-closure determination is the synthesis of linear loops. 
While computing the Zariski closure of an orbit is used to compute polynomial invariants of loops,
the loop synthesis problem reverses this direction: given an input polynomial invariant, one seeks to construct update rules for loop variables that maintain the invariant.

Suppose a deterministic polynomial loop involves integer variables $x,y,z,w$, and in each iteration, both $w$ and $y$ are incremented by one. The goal is to determine suitable updates to~$x$ and $z$ so that a specified invariant, such as $x^2-y^2z^2+z^3=0$, remains valid after each iteration, assuming that it holds initially.
In other words, we have a partially defined loop
    \begin{algorithmic}
 \Ensure $x^2-y^2z^2+z^3=0$
\While{$(*)$}%
 \State \(x \leftarrow P(x,y,z,w)\);
 \State \(y \leftarrow y +1\);
 \State \(z \leftarrow Q(x,y,z,w)\);
 \State \(w \leftarrow w+1\);
\EndWhile    
    \end{algorithmic}
\noindent
where $*$ denotes a wildcard that nondeterministically evaluates to false or true,
and the task is to determine assignments \(x\leftarrow P(x,y,z,w)\) and \(z\leftarrow Q(x,y,z,w)\) such that 
$x^2-y^2z^2+z^3=0$ is a loop invariant.
A valid solution in this case is to set \(P(x,y,z,w) \coloneqq y(y^2-w^2)\) and \(Q(x,y,z,w)\coloneqq y^2 - w^2\).
This kind of synthesis task is closely connected to the classical algebraic geometry problem of parametrising varieties defined by polynomial equations~\cite{aitelmanssour2025loops}.

Recent developments have focused on techniques for synthesising deterministic linear loops, often by encoding the synthesis task as a constraint-solving problem over algebraic structures~\cite{humenberger2020algebra,humenberger2022algebra, kenison2023polynomial, hitarth2024quadratic,aitelmanssour2025loops}.  
Geometrically, this amounts to finding an infinite orbit of a cyclic matrix group that remains entirely within the variety defined by the invariant. Some formulations restrict attention to infinite orbits in order to avoid trivial cases where the synthesis reduces to solving a finite system of polynomial equations~\cite{hitarth2024quadratic}.

Several concrete approaches have been proposed to instantiate these synthesis procedures. One line of work develops constraint-solving methods that generate loops preserving a given polynomial invariant, based on templates supplied by the user~\cite{humenberger2020algebra,humenberger2022algebra}. Other approaches have focused on restricted classes of invariants, such as those defined by a single quadratic equation~\cite{hitarth2024quadratic}  or by ideals generated by pure binomials~\cite{kenison2023polynomial}. In the latter case, the invariant variety is a union of toric varieties, and the synthesis relies on a construction by~\cite[Proposition 14]{galuppi2021toric}, showing that for every  toric variety \(V\) one can construct a  rational matrix \(M\) 
such that the Zariski closure of $\{M^n: n\in \mathbb{N}\}$ equals~$V$.
Using our terminology, toric varieties are topologically 1-generated.

A recent direction studies a strong variant of loop synthesis~\cite{aitelmanssour2025loops}, where the objective is not merely to keep the orbit within the target variety, but that the orbit closure of the loop variables coincides exactly with the input variety. 
This corresponds computationally to an alternative formulation of the orbit-closure determination problem studied in this paper, where the goal is to decide whether a given variety arises as the orbit closure of some vector under a matrix group.

The approach in~\cite{aitelmanssour2025loops} 
is restricted to cyclic groups generated by a single rational matrix (modelling single-path loops) and assumes that the input includes an explicit bound on the bit-size of the generator. In contrast, the present work solves a broader version of the strong loop synthesis problem, allowing for an arbitrary number~$s$ of commuting generators, without requiring any a priori bit-size bound. 
To relax the bit-bound restriction, we formulate the problem over the field~$\overline{\QQ}$ of algebraic numbers. 
This widens the scope from deterministic to nondeterministic loops.
{For instance, given a variety \(Z\subseteq \QQbar^d \), our methods can determine whether $Z$ is the orbit closure of a vector under a 2-generated matrix semigroup.  In other words, whether there exist matrices \(M_1, M_2\in\GL_d(\QQbar)\) and a vector \(\bm{v}\in\QQbar^d\) for which \(Z = \overline{\langle M_1,M_2\rangle \cdot \bm{v}}\).
This problem corresponds to that of synthesising a nondeterministic loop of the form
\begin{algorithmic}
\State \(\bm{x} \leftarrow \bm{v}\);
\While {($\ast$)}
\If {($\ast$)}
 \State \(\bm{x} \leftarrow M_1 \bm{x}\);
  \Else
  \State \(\bm{x} \leftarrow M_2 \bm{x}\);
  \EndIf
\EndWhile
\end{algorithmic}
such that the orbit closure of the loop coincides exactly with the input variety \(Z\).
}

\begin{example} \label{ex:loopGenintro2} Let us revisit \cref{ex:loopGenintro} from the viewpoint of loop synthesis. Recall that 
\(Q_1 = x_2^2 - x_1 - x_4\) and \(Q_2 = -2 x_4 x_2 - 2x_3^2 - \frac{1}{5} x_2 x_3\).
Consider the following task, \emph{construct a deterministic linear loop for which \(Q_1(\bm{x}) = 0 \land Q_2(\bm{x})=0\) is an invariant}.
This task amounts to determining an initial assignment \(\bm{x}\leftarrow \bm{v}\) and 
linear update \(\bm{x}\leftarrow M \bm{x}\) in the following loop:
\begin{algorithmic}
\State \(\bm{x} \leftarrow \bm{v}\);
\While {($\ast$)}
 \State \(\bm{x} \leftarrow M \bm{x}\);
\EndWhile
\end{algorithmic}
such that $Q_1(\bm{x})=0\wedge Q_2(\bm{x})=0$ is a loop invariant.  
In other words, we require that the following  Hoare triples are satisfied:
\begin{enumerate}
    \item $\{ \mathrm{true} \} \; \bm{x}\leftarrow \bm{v} \; \{Q_1(\bm{x})=0\wedge Q_2(\bm{x})=0 \}$;
    \item $\{Q_1(\bm{x})=0\wedge Q_2(\bm{x})=0\}\; \bm{x} \leftarrow M \bm{x} \; \{Q_1(\bm{x})=0\wedge Q_2(\bm{x})=0\}$.
\end{enumerate}

 Note that the above synthesis task may admit multiple solutions, including trivial solutions in which $M$ is the identity matrix.
 However, there is a natural notion of a maximal (or most permissive) solution, which
is a loop for which $Q_1(\bm{x})=0\wedge Q_2(\bm{x})=0$ is the strongest invariant. 
Geometrically, a solution is maximal if the set $Z$ of common zeros of $Q_1$ and $Q_2$ is the Zariski closure of the 
set $\langle M\rangle \cdot \bm{v}$ of reachable values of the program variables.  A motivation to look for such a solution is that 
it enables one to establish liveness properties, guaranteeing that certain program configurations can be reached.
(A sufficient condition for the orbit $\langle M\rangle \cdot \bm{v}$ to intersect a target set $Y$ is that 
the Zariski closure $Z$ meet the interior of $Y$.)
The algorithm we present in \cref{pro:loopGen} finds a most permissive solution.

The requirement that the linear loop be deterministic (have a single-path) translates, in the algebraic setting, to the requirement 
that \(Z\) be the orbit closure of a 1-generated algebraic matrix group. 
In this example, we instantiate the loop above with e.g.\ 
\[
M =  \begin{pmatrix}
    25& \phantom{-}0& -1& 20\\
    0& \phantom{-}5& \phantom{-}0& 0 \\
    0& -\frac{1}{2}& \phantom{-}5& 0\\
    0& \phantom{-}0& \phantom{-}1& 5
\end{pmatrix} \quad \text{ and } \quad \boldsymbol{v}=  \begin{pmatrix}
    1\\
    1\\
    0\\
    0
\end{pmatrix}
\]
and so deduce that \(Z = \overline{\langle M \rangle \cdot \bm{v}}\).
\end{example}

\subsection{Overview of the Main Result}

We reduce our determination problems to satisfiability problems in fragments of the first-order theory of  the algebraically closed  or  real-closed fields.
These theories are formulated in the first-order language of rings, which includes constant symbols $0$ and $1$, as well as binary function symbols for addition and multiplication. The theory of algebraically closed field is the set of all sentences in this language that are true over~$\QQbar$.
Similarly, the  theory of the real numbers consists of those sentences that are true over~$\mathbb{R}$.  
The following theorem gives a complexity bound on the decision problem for this theory.

\begin{theorem}[\cite{chistov1984complexity,basu1996quantifier,basu2006algorithms}]
    Consider a first-order sentence in the language of rings that mentions 
    $m$ polynomials in $d$ variables, with total degree at most $b$, and with $k$ quantifier alternations.  The truth of such a sentence over $\QQbar$ and $\QQbar \cap \mathbb{R}$ can both be decided in time $(mb)^{d^{2k+2}}$. 
\label{thm:FO}
 \end{theorem}

Following~\cite[Remark~13.10]{basu2006algorithms}, 
the truth of first-order sentences, over both $\QQbar$ and $\QQbar \cap \mathbb{R}$, 
with a fixed number of alternations can be decided in space~$(d \log{b})^{O(1)}$. 
 
The following theorem is our main  contribution: 

\begin{restatable}{theorem}{bigtheorem}
\label{pro:loopGen}
The orbit-closure determination problem for commutative  matrices with  parameters \((s,d,m,b)\) can be decided in time 
   $(mb)^{\mathrm{poly}(d)}$, and in space bounded by~$(d \log{b})^{O(1)}$. 
\end{restatable}

The proof of Theorem 2 shows how, given a variety $Z\subseteq \QQbar^d$,
to recover an $s$-generated group $G\subseteq \GL_d(\QQbar)$ and vector $\boldsymbol v\in Z$ such
that $Z=\overline{G \cdot \boldsymbol v}$.  The following is a summary of the key elements.  We will explain technical terms and expand on
each point immediately below:
\begin{enumerate}
\item
With respect to a suitable basis, the vector $\boldsymbol v$ has all
its entries 0 or 1 and hence can be guessed in nondeterministic polynomial time.
Crucially, given this form for $\boldsymbol v$, the orbit closure $Z$
itself carries a group structure under component-wise multiplication.
\item 
The semisimple part $G_s$ of $G$ is characterised up to isomorphism by an
additive subgroup $\Lambda$ of $\mathbb Z^k$ for some positive integer
$k$ such that the quotient group~$\mathbb Z^k/\Lambda$ has $s$ generators.  Moreover, the generators of $\Lambda$ have bit-size bounded polynomially in 
the description of $Z$; hence $\Lambda$ can also be guessed
in nondeterministic polynomial time.
\item
The existence of an $s$-generated unipotent group $G_u$ such that $Z = \overline{G_u\cdot G_s \cdot \boldsymbol v}$
can be expressed in first-order logic given descriptions of $Z$, $\boldsymbol v$, and  $\Lambda$.
This relies on the fact that for a unipotent matrix $U$, the matrix exponential $U^n$ is a matrix of polynomials in $n$.
\end{enumerate}

We write  $\textrm{Id}_d$ for the identity matrix of dimension~$d$. A matrix \(M\in \QQbar^{d\times d}\) is  \emph{unipotent} if \((M -\textrm{Id}_d)^d\) is the zero matrix, and 
\emph{semisimple} if it is diagonalisable over~\(\QQbar\).
Let $G$ be a commutative algebraic group.  By~\Cref{fact:decomp}, the subset $G_s$ of semisimple  matrices in $G$  forms an algebraic subgroup; similarly,  the set~$G_u$ of unipotent matrices in $G$ forms an algebraic subgroup. 
Moreover, the decomposition $G=G_u \cdot G_s$ holds in the commutative setting.

For a given  instance of the orbit-closure determination problem, let $Z$ be the zero locus of the input polynomials. If $Z$ is
the orbit closure  of a vector~$\bm{v}$
under the action of some group~$G$,  by the above, it can be written as
$\overline{G_u \cdot G_s \cdot \bm{v}}$.
If both  $G_s$ and $G_u$ are $s$-generated, then
by~\cref{Theremark},  the group $G$ is also $s$-generated.

A key technical result, \cref{lemma:MatrixUlike}, states that for given a commutative  matrix group~$G$, generated by matrices~$M_1,\ldots,M_s$ and  a vector~$\bm{v}$,  we can choose a change of basis (given by an invertible matrix~$P$)  that brings both the group~$G$ and the vector~$\bm{v}$ into a simplified and canonical form. In this new basis:
\begin{itemize}
    \item the semisimple subgroup  becomes diagonal;
    \item the unipotent subgroup becomes upper unitriangular; and
    \item   $\bm{v}$ is mapped to a binary vector $T\bm{1}$, where $T\in \{0,1\}^{d \times k}$ satisfies $T^{\top}T=\textrm{Id}_k$.
\end{itemize}
Our algorithm in~\cref{pro:loopGen}
guesses the matrix~$T$.
The algorithm proceeds by assuming  that~${P\bm{v}:=T\bm{1}}$. 
Write
\[PG_sP^{-1}=\overline{\langle  D_i : 1\leq i \leq s \rangle} \,.\]
In this new basis, the orbit-closure of~$P\bm{v}$ under 
$\langle  D_i : 1\leq i \leq s \rangle$ forms a linear algebraic group, and  becomes  
the zero set of a pure binomial ideal~$I$.
Recall that a  pure binomial ideal 
	in the variables~$\bm{x}=(x_1,\ldots,x_d)$
	is an ideal generated by polynomials of the from~$\bm{x}^{\bm{\lambda}}-\bm{x}^{\bm{\lambda'}}$, where  
	$\bm{\lambda}, \bm{\lambda'} \in \NN^d$.
	Associated with any such ideal is its exponent lattice~$\Lambda$, a subgroup of $\mathbb{Z}^d$ collecting the vectors~$\bm{\lambda} - \bm{\lambda'}$.
Studying this lattice provides key structural insight into the underlying variety defined by~$I$.

By the requirement on the number of generators, the orbit closure
 $\overline{\langle  D_i : 1\leq i \leq s \rangle \cdot P\bm{v}}$ seen as a linear algebraic group must also have $s$  topological generators. 
By~\cref{prop:cyclic}, 
 the torsion subgroup of 
     $\mathbb{Z}^d/\Lambda$ has at most $s$  generators.
A careful analysis in~\Cref{claim:intersection}, combined with~\cref{prop:bound}, shows that the degree bound~$b$ on the defining polynomials of~$Z$ also applies to the generators of the binomial ideal~$I$.
The degree bound~$b$ allows the algorithm to guess a lattice $\Lambda$, generated by vectors with entries bounded in absolute value by~$b$, and such that $\mathbb{Z}^d/\Lambda$ has at most $s$
  generators, in line with~\cref{prop:cyclic}.
For this lattice $\Lambda \subseteq \mathbb{Z}^d$, we define $H_\Lambda$ as a subgroup of the \(d\)-dimensional multiplicative variety, where $\boldsymbol{a} \in H_\Lambda$ if $\bm{a}^{\bm{\lambda}}=1$ for all $\bm{\lambda}\in \Lambda$.

 By~\cref{Theremark}, there exist unipotent matrices 
$U_1, \ldots, U_s$ that topologically generate  
the unipotent subgroup $G_u$ of~$G$. Furthermore,  by~\cref{prop:uni}
the following equality holds:
\[G_u=\biggl\{\exp\biggl(\sum_{i=1}^s t_i \log U_i\biggl)  :  t_1,\ldots,t_s\in\QQbar\biggr\}.\] 
Since the power series of $\exp$ has at most $d$ terms in this case (see~\Cref{sec:preliminaries}), it follows that 
the group $G_u$ is definable by a first-order formula with parameters~$U_1, \ldots,U_s$. 

We are now in a position to encode the given instance of the problem as a first-order sentence in the theory of  real-closed fields. An important observation enabling this reduction is that, for linear algebraic groups, the Zariski and Euclidean closures of an orbit coincide (see~\Cref{fact:EucZar}). In our encoding, we rely on this fact to describe the algebraic closure of a constructible set using Euclidean conditions:
by requiring that for all $\varepsilon > 0$, there exists a group element constructed by our guesses mapping the image within $\varepsilon$ of some  vector in the input~$Z$.

More in detail, the orbit-closure verification task is  expressed    as a first-order sentence with quantifier prefix of the form~$\exists^* \forall^* \exists^*$, 
that is, a block of existential quantifiers followed by a block of universal quantifiers and a final existential quantifier.
The outermost existential quantifiers encode the possible choices of the matrices $P$ and $U_1, \ldots, U_s$, while the
equality of $PZ$ and 
\begin{equation*}
 \overline{\langle U_i :1\leq i\leq s \rangle \cdot TH_\Lambda}
\end{equation*}
is encoded by a $\forall^*\exists$-sentence with parameters $P$ and $U_1, \ldots, U_s$.
The algorithm returns "yes", meaning that $Z$ is  
an orbit closure  of a vector $\bm{v}$
under the action  of the group~$G$, if the above sentence is satisfiable over~$\mathbb{R}$. 
By \cref{thm:FO}, the truth of such a sentence can be decided in time~$(mb)^{\mathrm{poly}(d)}$.
Then the overall complexity bound follows from the fact that the number of choices of the lattice $\Lambda$ and vector $P\bm{v}$ is at most $(2b)^{d^2+1}$.
This concludes our informal overview of the proof of \cref{pro:loopGen}; 
the detailed proof can be found in \cref{sec:orbitdet}. 

\noindent \textbf{Orbit-Closure vs. Group Determination.} En route to proving~\cref{pro:loopGen} on orbit-closure determination, we consider a  variant---namely  group determination. 

For group determination  we first consider a simpler setting where the input polynomial ideal~$I\subseteq \mathbb Z[\bm{x}]$ is a pure binomial ideal.
Recall that we can associate a lattice 
$\Lambda = \{ \bm{\lambda} - \bm{\lambda'} \in \mathbb Z^d : \bm{x}^{\bm{\lambda}}-\bm{x}^{\bm{\lambda'}} \in I\}$ with~$I$. 
 It is well-known that, if the input ideal $I$ defines a group~$G$ then it is necessarily topologically generated by diagonal matrices. 
If the torsion subgroup of 
     $\mathbb{Z}^d/\Lambda$ is $s$-generated, then, by~\cref{prop:cyclic},
the minimal number of  topological generators of $G$ is either $s$ or $s+1$ (depending on the rank of $\Lambda$).
However, in the  setting of orbit-closure determination, %
this lower bound  on the  number of generators  may no longer hold, as shown by the following example.

\begin{example}\label{ex:grouporbit}

    Let  $\Lambda \subseteq \Z^3$ be the associated lattice of the pure binomial ideal $I \coloneqq \langle x_{1}^2 - x_{3}^2 , x_{2}^2 - x_{3}^2 \rangle$.
    That is, $\Lambda = \{c_1 (2,0,-2)^\top + c_2(0,2,-2)^\top : c_1, c_2 \in \ZZ\}$.
    By the fundamental theory of finitely generated abelian groups, the quotient group $\Z^3/\Lambda$ is isomorphic to $(\Z/2\Z) \times (\Z/2\Z) \times \Z$.
    In particular, the torsion subgroup of $\Z^3/ \Lambda$ is $(\Z/2\Z) \times (\Z/2\Z)$, a product of two cyclic groups, which is~$2$-generated. 

Here and in the following, let $\diag{x_1,x_2,x_3} \in \GL_3(\QQbar)$ denote a diagonal matrix with~$x_1, x_2, x_3$ (in this order) on the diagonal. 
Let 
$G \leq \GL_3(\QQbar)$ be the algebraic subgroup of diagonal matrices induced by \(I\), that is,
\[G = \{\diag{\bm{a}} : P(\bm{a}) = 0 \text{ for all } P\in I\}.\]
Observe that $G$ is isomorphic to $H_\Lambda$ (recall that the elements of the latter are $\bm{a}$ such that $P(\bm{a}) = 0$). 
By ~\cref{prop:cyclic},
the minimal number of topological generators for~$G$ is 2.
Indeed, $G$ is topologically generated by the diagonal matrices~\(\diag{2,-2,2}\) and \(\diag{-2,2,2}\).

Changing our viewpoint, let us consider $I$
as defining a subvariety $V$ of~$\QQbar^3$ rather than a set of diagonal matrices in $\GL_3(\QQbar)$.
The variety~$V$ can be written as the orbit closure of a vector $\bm{v}$ 
   under the action of the 1-generated (cyclic) group \(\overline{\langle M\rangle}\), where  
    \[M = \begin{pmatrix}
        0& -2& 0\\
        2 & \phantom{-}0 & 0 \\
        0 & \phantom{-}0 & 2
    \end{pmatrix} \quad \text{and} \quad \bm{v} = \begin{pmatrix}
        1\\
        1\\
        1
    \end{pmatrix}.
    \]
\hfill $\blacktriangleleft$
\end{example}

\subsection{Outline and Structure}
The remainder of the paper is structured as follows.  
In~\cref{sec:preliminaries} we give useful notations and definitions (we refer the reader to the Appendix for extended preliminaries).
In \cref{sec:group} we present procedures for the group determination problem (\cref{prop:semisimplematrix,prop:invertiblematrix}).
In \cref{sec:orbitdet} we present procedures for the orbit-closure determination problem (\cref{prop:loopGensemi,pro:loopGen}).
Finally, in \cref{sec:generators} we present two
algorithms that compute a generator in the special case that the input variety is the Zariski closure  of a cyclic group.

\section{Primer on Algebraic Geometry}
In this section we introduce basic notions from algebraic geometry that are needed in the paper. 
We refer the reader to~\cite{cox2015ideals} for
further details and examples.

\label{sec:preliminaries}

\noindent \textbf{Ideals and Varieties.}
Recall that an algebraic number is a complex number that is the root of a univariate polynomial with integer coefficients.
The collection of all algebraic numbers forms a sub-field of $\mathbb C$, which is denoted 
\(\QQbar\).
We write $\QQbar[\bm{x}]$ for the
ring of polynomials in the variables~$\bm{x}=(x_1,\ldots,x_d)$ with coefficients in $\QQbar$.  
A \emph{polynomial ideal}~$I$ is an additive subgroup of~$\QQbar[\bm{x}]$
that is closed under multiplication in $\QQbar[\bm{x}]$.  
Given a  set of polynomials \(S  \subseteq \QQbar[\bm{x}]\), we denote by \(\langle S\rangle\) the ideal \emph{generated by \(S\)}:
\[\langle S \rangle \coloneqq \left\{ P_1 Q_1 + \cdots + P_k Q_k : k \in \NN, Q_1, \ldots, Q_k \in \QQbar[\bm{x}], P_1, \ldots, P_k \in S \right\}.\]
For example, $\langle x_1 \rangle$ is the ideal of all polynomials that are multiples of $x_1$.

An \emph{algebraic set} (or \emph{variety}) is the set of common zeroes of a finite collection of polynomials.
We may also refer to an algebraic set as the \emph{zero locus} of these polynomials.
By Hilbert's basis theorem every polynomial ideal \(I \subseteq \QQbar[\bm{x}]\) is finitely generated, that is, there exists a finite set \(S  \subseteq \QQbar[\bm{x}]\) such that  \(I = \langle S\rangle\).
Thus the zero set of~$I$
    \begin{equation*}
       V(I) \coloneqq   \{\bm{a}\in\QQbar^d : P(\bm{a})=0 \text{ for all } P\in I\} 
       \end{equation*}
is a variety. As an example, consider an ideal~$I$ of the polynomial ring~$\QQbar[x_1, x_2]$ generated by polynomials~$\{x_1^2, x_2^2\}$.
The ideal 
\[I = \langle x_1^2, x_2^2 \rangle = \{x_1^2 Q_1(x_1, x_2) + x_2^2 Q_2(x_1, x_2) : Q_1, Q_2 \in \QQbar[x_1,x_2]\}\]
contains polynomials such as $x_1^2$, $x_1^2 - \sqrt{2}x_2^2$, and $x_1^3+x_1^2x_2 + x_1x_2^2 - x_2^3$.
The zero set of $I$ 
is the variety
\[V(I) \coloneqq  \{(a_1,a_2) \in\QQbar^2 : P(a_1,a_2) = 0 \text{ for all } P\in I\} = 
\{(a_1,a_2) \in\QQbar^2 : a_1^2 = 0 \wedge a_2^2 = 0\}.\]
Dually, the set of polynomials that vanish on a set $E \subseteq \QQbar^d$, denoted
\[
    I(E)\coloneqq \{P \in \QQbar[\bm{x}] :  P(a) = 0 \text{ for all } a \in E\},
\]
is an ideal in $\QQbar[\bm{x}]$ and is referred to as the
\emph{vanishing ideal} of $E$.
Resuming our example, we can consider $E = V(I)$ for $I = \langle x_1^2 , x_2^2 \rangle$. The vanishing ideal of this variety is
\begin{multline*}
I(V(I))\coloneqq \{P \in \QQbar[x_1,x_2] :  P(a_1, a_2) = 0 \text{ for all } (a_1,a_2) \in V(I)\}
\\ =\{P \in \QQbar[x_1,x_2] :  P(a_1, a_2) = 0 \text{ for all } (a_1,a_2) \text{ such that } a_1^2 = a_2^2 = 0\} = \langle x_1, x_2 \rangle.
\end{multline*}
The ideal $I(V(I))$ may, in general, be different from $I$.

\noindent \textbf{Zariski Topology.}
A \emph{topology} on a set~$X$ is a collection~$\tau$ of subsets of~$X$, called \emph{closed sets}, that satisfy the following axioms:
\begin{enumerate}
    \item $\varnothing, X \in \tau$;
    \item An arbitrary intersection of a collection of sets in~$\tau$ belongs to~$\tau$;
    \item Any finite union of sets in~$\tau$ belongs to~$\tau$.
\end{enumerate}

\color{black} The \emph{Zariski topology} on 
 \(\QQbar^d\) has as its closed sets 
 the varieties in \(\QQbar^d\).
Given a set~\(E\subseteq \QQbar^d\), we denote by~\(\overline{E}\) the \emph{closure} of~\(E\) in the Zariski topology, that is, 
the smallest algebraic set that contains~\(E\). 
{We say that a set $E \subseteq \QQbar^d$ is dense in an algebraic set $V$, if~$\overline{E} =V$.}

A closed set $A \subseteq \QQbar^d$ is \emph{irreducible} if it cannot be written as the union of two proper closed subsets.  A maximal irreducible closed subset of $A$ is called 
an \emph{irreducible component} of $A$. 
By {Hilbert's basis theorem}
every closed set $A$ can be written as a finite union of its irreducible components.
For example, the set $A = \{ \boldsymbol a \in \QQbar^2 : a_1a_2=0 \}$ has irreducible components 
$A_1 = \{ \boldsymbol a \in \QQbar^2 : a_1=0\}$ and $A_2 = \{ \boldsymbol a \in \QQbar^2 : a_2=0\}$.
Given a closed set~$E$, denote by $\dim E$ the dimension of the variety~$E$, that is, 
the maximal length of a strictly decreasing  chain of nonempty irreducible subvarieties of $E$. We define the dimension of an arbitrary set to be the dimension of
its closure. 

Given an ideal $I\subseteq \QQbar[X]$ over the variables  $X=\{x_{i,j}\}_{1\leq i,j\leq d}$, and matrix $M \in \QQbar^{d\times d}$, we write 
$M\cdot I$ for the ideal $\{P(MX) \in \QQbar[X] : P \in I\}$.
Clearly, $V(M\cdot I)=\{ A \in \QQbar^{d\times d}: MA \in V(I) \}$.

\noindent \textbf{Linear Algebraic Groups.}
A matrix \(M\in \QQbar^{d\times d}\) is  
\emph{nilpotent} if \(M^n=0\) for some $n\in \NN$. 
It is
\emph{unipotent} if \(M - \textrm{Id}_d\) is nilpotent, and 
\emph{semisimple} if it is diagonalisable over~\(\QQbar\).
 The matrix \(M\in \QQbar^{d\times d}\) is called 
 \emph{upper triangular}  if all entries below the main diagonal are zero.
We use the term \emph{upper unitriangular} to refer to an upper triangular matrix whose entries along the main diagonal are all ones.

Write  $\GL_d(\QQbar)$
for the group of  invertible $d\times d$ matrices with entries in $\QQbar$.
We identify $\mathrm{GL}_d(\QQbar)$ with the variety 
\[\bigl\{(M,y)\in \QQbar^{d\times d}\times \QQbar: \det(M)\cdot y=1\bigr\}\,.\]
Under this identification, matrix multiplication is a polynomial map $\GL_d(\QQbar)\times \GL_d(\QQbar) \rightarrow \GL_d(\QQbar)$, and, by Cramer's rule, matrix inversion is also a polynomial map $\GL_d(\QQbar) \rightarrow \GL_d(\QQbar)$. 
A \emph{linear algebraic group} (or algebraic matrix group) $G$ is a Zariski-closed subgroup of $\GL_d(\QQbar)$.
In other words, a subgroup~$G \leq \GL_d(\QQbar)$ is algebraic if it can be defined by polynomial equalities. 
As an example take the subgroup of diagonal matrices in $\GL_d(\QQbar)$,  defined by 
\[\{M =(m_{ij})_{1\leq i, j \leq d} \in \GL_d(\QQbar): \; \; m_{ij} = 0 \text{ for all } i \neq j \}.\]

Denote by  $G_s$  the subset of semisimple  matrices in $G$, and by $G_u$ the subset of unipotent matrices. The following fact is well-known, see for example~\cite[Chapter 6]{HumphreysLAG}.

\begin{fact}
\label{fact:decomp}
    For a commutative algebraic group \(G\leq \GL_d(\QQbar)\), the algebraic subgroups \(G_s\) and \(G_u\) form algebraic subgroups {of $G$}; moreover, we have the decomposition of \(G\) into $G_u \cdot G_s$.
\end{fact}

The next fact follows from Chevalley's Theorem; see~\cite[Chapter AG, Corollary 10.2]{borel1991} and~\Cref{app:proofs} for details. 
\begin{restatable}{fact}{factEucZar}
\label{fact:EucZar}
    Let $G \leq \GL_d(\QQbar)$ be an  algebraic group and $\bm{v}\in \QQbar^d$ a vector, the Zariski and Euclidean closures of the orbit $G \cdot \bm{v}$ coincide.
\end{restatable}

We say that $G$ is \emph{topologically generated by}
$S\subseteq \GL_d(\QQbar)$ if $G$ is the smallest Zariski-closed subgroup of 
$\GL_d(\QQbar)$ that contains $S$, that is, $G=\overline{\langle S\rangle}$.
If $G$ is topologically generated by a set with $s$ elements then we say that is
$G$ is \emph{$s$-generated}.

\begin{restatable}{proposition}{propgen}
	\label{prop:gen}
	Let $G \leq \GL_d(\QQbar)$ be an algebraic group, then $G$ is topologically generated by a finite set of matrices.
\end{restatable}
The proof of \cref{prop:gen} is given in \cref{app:proofs}.

\noindent \textbf{Lattices and the Multiplicative Group $\mathbb{G}_m^d$.}
The \emph{rank} of an abelian group \(\Lambda\) is the size of a maximal linearly independent subset~\cite{Lang2002}. 
A subgroup \(\Lambda\subseteq \Z^d\) is called a \emph{lattice} (and has rank at most~$d$).
The set~$\Z^d/\Lambda$ contains all cosets of~$\Lambda$ in~$\Z^d$, that is,
sets of the form $\bm{g} + \Lambda = \{\bm{g}+\bm{v} : \bm{v} \in \Lambda\}$ where $\bm{g} \in \Z^d$. 
Together with an additive operation, the set of cosets defines the~\emph{quotient group}~$\Z^d/\Lambda$.
The \emph{torsion subgroup} of~$\ZZ^d / \Lambda$ is the subgroup of $\ZZ^d / \Lambda$ consisting of all its elements of finite order.

The $d$-dimensional multiplicative group 
over $\QQbar$ is defined as
 \begin{gather}
 \mathbb{G}_m^d = \mathbb{G}_m^d(\QQbar) \coloneqq   \left\{\boldsymbol{a} \in \QQbar^{d} : a_1\cdots a_d
   \neq 0 \right\}.
   \label{eq:MULT-INEQ}
 \end{gather}
   Here the subscript $m$ stands for \emph{multiplicative}.
Evidently, this is a commutative group with respect to pointwise multiplication.
We identify $\mathbb{G}_m^d$ with the subgroup of diagonal matrices in $\GL_d(\QQbar)$ via the map $\Delta$
that sends $(a_1,\ldots,a_d) \in \mathbb{G}_m^d$
to the diagonal matrix $\Delta(a_1,\ldots,a_d) \in \GL_d(\QQbar)$
which has $a_1, \dots, a_d$ (in this order) on the diagonal and zeros elsewhere.

We can now explain two important ingredients of our results.
We have so far identified the invertible diagonal matrices with the algebraic group~$\mathbb{G}_m^d(\QQbar)$.
Now, we define the necessary vocabulary to discuss the correspondence of algebraic subgroups of~$\mathbb{G}_m^d(\QQbar)$ and lattices in~$\ZZ^d$. 
This comes in handy when we need to decide whether a certain subgroup of diagonal invertible matrices is~$s$-generated, in which we reduce it to a lattice problem.

Given a lattice $\Lambda \subseteq \mathbb{Z}^d$, define
\[H_\Lambda \coloneqq   \{ \boldsymbol{a} \in \mathbb{G}_m^d : \enspace \forall
   \boldsymbol{v}\in \Lambda. \; a_1^{v_1}\cdots a_d^{v_d}=1 \}.\]%
   The map $\Lambda \mapsto H_{\Lambda}$ is an isomorphism
 between lattices and algebraic subgroups
 of $\mathbb{G}_m^d$.  
 This implies that $\mathbb{G}_m^d$ is topologically generated by any $d$-tuple $(g_1,\ldots,g_d)$ of multiplicatively
independent elements of $\QQbar$.

For variables $\bm{x}=(x_1,\ldots,x_d)$, write $\bm{x}^{\bm{\lambda}}$ for the monomial $\prod_{i=1}^d x_i^{\lambda_i}$. A \emph{pure binomial ideal} is 
an ideal generated by polynomials of the form $\bm{x}^{\bm{\lambda}}-\bm{x}^{\bm{\lambda'}}$, where $\bm{\lambda}$ and $\bm{\lambda'}$ are non-negative integer vectors.  More generally, a \emph{binomial ideal} is one that is generated by polynomials of the form~$\bm{x}^{\bm{\lambda}}-\theta\bm{x}^{\bm{\lambda'}}$, where $\theta \in \QQbar$.
It is known that
 the vanishing ideal 
$I\subseteq \QQbar[\bm{x}]$ of an algebraic subgroup of 
$\mathbb{G}_m^d$ is a pure binomial ideal.

\begin{example} \label{ex:lattice}

Let $\Lambda \coloneqq \{(2n, -n) : n \in \ZZ \} = (2,-1)\cdot \ZZ$ be a lattice in $\ZZ^2$ of rank~1. 

A vector $(a,b) \in \ZZ^2$ is equivalent to $(a+2b,0)$ modulo $\Lambda$.  Hence 
the quotient group~$\Z^2/\Lambda$ comprises a set of cosets $\{(c,0) + \Lambda: c \in \Z\}$ with group operation defined by componentwise addition. 
Thus $\Z^2/\Lambda$ is isomorphic to the additive group~$\ZZ$, which has trivial torsion subgroup since~$0$ is the only element of finite order.

We have \begin{align*}
    H_\Lambda = &  \{ (x,y) \in \mathbb{G}_m^2 : \enspace \forall
   n \in \ZZ. \; x^{2n}y^{-n}=1\} \\= 
   & \{ (x,y) \in \mathbb{G}_m^2 : \enspace x^2=y\}.
\end{align*}

   The vanishing ideal of~$H_\Lambda$ is pure binomial, generated by $x^2-y$.
\hfill $\blacktriangleleft$
\end{example}

The following proposition shows how to recover the generators of the lattice
and hence the pure binomial ideal that vanishes on an algebraic subgroup of $\mathbb{G}_m^d$ 
defined by its equations.
 \begin{proposition}[{\cite[Proposition~3.2.14]{BombieriGubler}}]
     Let $G$ be a subgroup of $\mathbb{G}_m^d$ defined by polynomial equations $ \sum_{\bm{\lambda} \in \mathcal{L}_i }a_{i,\boldsymbol{\lambda}}\boldsymbol{x}^{\boldsymbol{\lambda}}=0$ for $i=1,\ldots,m$, where 
     $\mathcal{L}_i \subseteq \mathbb{Z}^d$.
     Then $G=H_\Lambda$, where $\Lambda \subseteq \mathbb{Z}^d$ is generated by
vectors of the form $\bm{\lambda}_i - \bm{\lambda'}_i$
with $\bm{\lambda}_i, \bm{\lambda'}_i \in \mathcal{L}_i$, for $i \in \{1, \dots, m\}$.
 \label{prop:bound}
 \end{proposition}

The next proposition follows from  a standard result in Diophantine geometry concerning the number of generators of a subgroup of $\mathbb{G}_m^d$ (cf.~\cite[Chapter 3]{BombieriGubler});
we give a proof in \cref{app:proofs}.

\begin{restatable}{proposition}{propcyclic}
\label{prop:cyclic}
 Let $\Lambda \subseteq \ZZ^d$ be a lattice of rank $r$. Then
\begin{enumerate}
     \item \label{itemprop:1}  the torsion subgroup of $\ZZ^d / \Lambda$ is $s_0$-generated, for some $s_0 \leq r$, and
     \item \label{itemprop:2} 
$H_\Lambda$ is $s$-generated, where $s\coloneqq s_0$ if $\Lambda$ has full rank and otherwise $s\coloneqq\max(s_0,1)$.
Furthermore, $s$ is the minimal number of topological generators for~$H_\Lambda$.
 \end{enumerate}
\end{restatable}

\noindent \textbf{Unipotent Matrices.}
For a $d\times d$ unipotent matrix $A$ and a $d\times d$ nilpotent matrix $B$, define 
\begin{equation*} \log(A) \coloneqq  \sum_{k=1}^{d-1} (-1)^{k+1}\frac{(A-\textrm{Id}_d)^k}{k} 
\quad \text{and} \quad
\exp(B) \coloneqq   \sum_{k=0}^{d-1} \frac{B^k}{k!}.
\end{equation*}
Let \(G\subseteq \GL_d(\QQbar)\) be a commutative subgroup of unipotent matrices.
The set  \(L\coloneqq   \{\log(A) : A\in G\}\) is a linear subspace of 
$\QQbar^{d^2}$
consisting of nilpotent matrices \cite[Chapter II, Section 7.3]{borel1991}. %
Moreover, \(\exp\colon L\to G\) and \(\log\colon G\to L\) yield polynomial isomorphisms between  \(L\) and   \(G\) as algebraic groups, with  additive and multiplicative group structures respectively.
Taken together, these observations lead to the following proposition.

\begin{proposition}
Let \(G\) be a  commutative group
of unipotent matrices and \(L\coloneqq\{ \log(A): A \in G\}\) the associated linear subspace of nilpotent matrices.
Then \(G\) has a topological generator of cardinality \(s\) if and only if \(L\) is spanned by a set of \(s\) matrices
as a $\QQbar$-vector space.
\label{prop:uni}
\end{proposition}
\begin{proof}
For all \(A_1,\ldots, A_s\in G\) we have the following equivalences:  
\begin{align*}
  & \, \{A_1,\ldots, A_s\} \text{ topologically generates }  G, \\
 \Leftrightarrow \, &\, \{A_1^{n_1} A_2^{n_2} \cdots A_s^{n_s} : n_1,\ldots, n_s \in \Z\} \text{ is dense in } G, \\
 \Leftrightarrow \, &\, \{\textstyle \sum_{i=1}^s n_i\log(A_i)  : n_1,\ldots, n_s \in \Z\} \text { is dense in } L, \\
 \Leftrightarrow \, &\, \{\textstyle \sum_{i=1}^s t_i\log(A_i) : t_1, \ldots, t_s\in\QQbar\} = L, %
\end{align*}
as desired.
\end{proof}

\begin{example}
We consider $2 \times 2$ matrices of the form
\[A_t \coloneqq \begin{pmatrix}
        1& t\\
        0 & 1
    \end{pmatrix} \, ,\]
    where $t \in \QQbar$. Since $(A_t - \textrm{Id}_2)^2 = 0$, the matrix $A_t$ is unipotent.
Moreover, $G \coloneqq \{A_t : t \in \QQbar \}$ is a (multiplicative) commutative subgroup of~$\GL_2(\QQbar)$.
The matrix \[
    \log(A_t) = A_t - \textrm{Id}_2 = \begin{pmatrix}
        0& t\\
        0 & 0
    \end{pmatrix}\]
is nilpotent for any~$t \in \QQbar$.
    It further holds that \[L \coloneqq \{\log(A_t) : t\in \QQbar\} = \left\{\begin{pmatrix}
        0& t\\
        0 & 0
    \end{pmatrix} : t\in \QQbar\right\}
    \]
    is a linear subspace of~$\QQbar^4$.
    In fact, $L$ is spanned by a single matrix over~$\QQbar$ and so, by \cref{prop:uni}, 
    $G$ is topologically 1-generated.
    Indeed, $G$ is the Zariski closure of, say, \[\langle A_1 \rangle  = 
    \left\langle \begin{pmatrix}
        1 & 1\\
        0 & 1
    \end{pmatrix} \right\rangle = 
    \left\{ \begin{pmatrix}
       1& n\\
        0 & 1
   \end{pmatrix} : n \in \ZZ \right\}.\]
The map $\exp$ is an isomorphism of $L$ and $G$:
\[
\exp : \begin{pmatrix}
        0& t\\
        0 & 0
    \end{pmatrix}
    \mapsto
    \textrm{Id}_2 + \begin{pmatrix}
        0& t\\
        0 & 0
    \end{pmatrix} = \begin{pmatrix}
        1& t\\
        0 & 1
    \end{pmatrix}.
\] 
\hfill $\blacktriangleleft$
\end{example}

\section{Commutative Group Determination}
\label{sec:group}

Recall that the group determination problem with parameters~$(s,d,m,b)$ asks,  given $s \in \N$ and a family 
 of $m$ polynomials in~$\QQ[\{x_{i,j}\}_{1\leq i,j\leq d}]$, each
 of total degree at most~$b$,
 to determine whether their zero locus $Z\subseteq\GL_d(\QQbar)$ is an $s$-generated algebraic matrix group.
In this section, we first demonstrate a procedure for this problem subject to the 
constraint that the underlying group is semisimple and commutative~(\cref{prop:semisimplematrix}).
Next, we  generalise this result by lifting the requirement that the matrices are semisimple~(\cref{prop:invertiblematrix}).

\begin{proposition}
 \label{prop:semisimplematrix}
The group determination problem for commutative semisimple matrices with   parameters $(s,d,m,b)$ can be decided
 in time 
   $(mb)^{\mathrm{poly}(d)}$, and in space bounded by~$(d \log{b})^{O(1)}$.
   \end{proposition}
   \begin{proof}
 Recall that the input to the problem consists of~$m$ polynomials, each of total degree at most~$b$, together with a natural number $s$.  Let $Z$ be the subvariety  of $\GL_d(\QQbar)$ defined by the input polynomials. 
   The task is to determine whether  $Z$ is a group that is topologically generated by at most $s$  commutative semisimple  matrices.

Suppose that the input is a positive instance of the problem, meaning that  
 $Z$ is an algebraic group that is topologically generated by 
 commutative semisimple  matrices~$M_1, \ldots, M_s \in \GL_d(\overline{\Q})$.   
  Recall that  commutative semisimple matrices are simultaneously diagonalisable~\cite[Theorem 1.3.21]{horn2012matrix}.
       Thus there exists $P \in \GL_d(\overline{\Q})$
       such that $D_i \coloneqq P^{-1} M_i P$ are diagonal.
       Let $G$ be the subgroup of
        $\mathbb{G}_m^d$ defined by 
        \[G\coloneqq \{ g\in \mathbb{G}_m^d : \Delta(g) \in \overline{\langle D_i: 1\leq i \leq s\rangle}\}\,.\]
Since 
$Z=\overline{\langle M_i: 1\leq i \leq s\rangle}$, we have
   $P^{-1}ZP =\{ \Delta(g) : g \in G\}$.
   Moreover, as 
    \(Z\) is the zero  set of polynomials of total degree at most \(b\) and   conjugation by the linear transformation \(P\) does not increase the total degrees of the transformed polynomials, it holds that 
   $P^{-1}ZP$ is the zero set
    of polynomials of degree at most $b$. 
    Hence,  
   by \cref{prop:bound}, 
    the group $G$ has the form $H_\Lambda$ for some lattice $\Lambda\subseteq \mathbb{Z}^d$
whose generators have norm at
most $b$  and such that $H_\Lambda$ is $s$-generated.

Conversely, for every such  lattice~$\Lambda$, the  variety $P \{ \Delta(g) : g \in H_{\Lambda}\} P^{-1} $ is topologically $s$-generated by commutative semisimple matrices.

The above reasoning shows the correctness of the following nondeterministic decision  
procedure:
\begin{enumerate}
    \item Guess a lattice $\Lambda\subseteq\mathbb{Z}^d$ whose generators have norm at most $b$ such that $H_\Lambda$ is  $s$-generated. 
    \item Output "yes" %
    if there exists $P \in \mathrm{GL}_d(\QQbar)$ such that 
    $P^{-1}ZP = \{ \Delta(g) : g \in H_\Lambda \}$.
\end{enumerate}

By~\cref{prop:cyclic}, Step 1 amounts to guessing~$\Lambda$ subject to the condition that the torsion subgroup of~$\ZZ^d/\Lambda$ is generated by at most~$s$ elements.
We can understand this observation in terms of the  number of non-unit elementary divisors of the desired lattice~$\Lambda$; see~\cref{app:proofs} for more details.

Step 2 amounts to checking the truth in $\QQbar$ of the sentence
\begin{multline*}
   \exists P \in \GL_d(\QQbar) \; \forall A = (a_{ij})_{1\leq i, j \leq d} \in \GL_d(\QQbar)  \;\\
\left(P^{-1}AP \in Z \Leftrightarrow \textstyle \bigwedge_{i\neq j} a_{ij}=0 \wedge  (a_{11}, a_{22},\ldots,a_{dd}) \in H_\Lambda\right),
\end{multline*}
with respect to the theory of algebraically closed fields.  
By \cref{thm:FO}, this can be done in time $(mb)^{\mathrm{poly}(d)}$.
The claimed running time for the overall procedure follows from
the fact that the number of possibilities for the lattice $\Lambda$ is at most~$(2b)^{d^2}$. 
   \end{proof}

By the following remark, a commutative algebraic group~$G$ is $s$-generated if its unipotent subgroup~$G_u$ and semisimple subgroup~$G_s$ are both $s$-generated. 

\begin{remark}
\label{Theremark}
    Let $U = \overline{\langle U_i: 1\leq i \leq s \rangle}$ be a commutative unipotent $s$-generated algebraic group and $S = \overline{\langle S_i :1\leq i \leq s \rangle}$ a semisimple commutative $s$-generated algebraic group. Suppose moreover that the matrices $\{U_i, S_i: 1\leq i \leq s \}$ are  commutative. 
    Then we have 
   \( \overline{ \langle U_i, S_i: 1\leq i \leq s  \rangle }= \overline{\langle S_iU_i: 1\leq i \leq s  \rangle}
   \).
   This relies on the fact that if
$A\in \GL_d(\QQbar)$ is semisimple and 
$B \in \GL_d(\QQbar)$ is unipotent, then both $A$ and $B$ lie in the Zariski closure of the subgroup generated by their product $AB$; see~\cite[Section 15.3]{HumphreysLAG}.
\end{remark}

The next proposition generalises the procedure described in \cref{prop:semisimplematrix}.
In \cref{prop:invertiblematrix} we consider the group determination problem for \(s\)-generated commutative algebraic groups.
Key to our generalisation is the determination of a matrix \(P\in\GL_d(\QQbar)\) and properties associated with the semisimple and unipotent subgroups of the group \(P^{-1}ZP\).

\begin{proposition} 
\label{prop:invertiblematrix}
The group determination problem for commutative matrices with  parameters \((s,d,m,b)\) can be decided in time 
   $(mb)^{\mathrm{poly}(d)}$, and in space bounded by~$(d \log{b})^{O(1)}$. 
 \end{proposition}

\begin{proof}
 Recall that the input to the problem consists of~$m$ polynomials, each of total degree at most~$b$, together with a natural number $s$.  Let $Z$ be the subvariety  of $\GL_d(\QQbar)$ defined by the input polynomials. 
   The task is to determine whether  $Z$ is a group that is topologically generated by at most $s$  commutative   matrices.

Suppose that the input is a positive instance of the problem, that is, 
$Z$ is an algebraic group that is topologically generated by 
 commutative   matrices~$M_1, \ldots, M_s \in \GL_d(\overline{\Q})$.   
  Recall that  commutative  matrices are simultaneously triangularisable~\cite[Theorem~2.3.3]{horn2012matrix}.
       Thus there exists $P \in \GL_d(\overline{\Q})$
   such that 
there exist diagonal matrices~$D_i$ and 
upper unitriangular matrices~$U_i$, 
$1\leq i\leq s$, where 
$P^{-1}M_iP = D_iU_i$, and moreover 
$D_i$ and $U_i$ commute.
Then we can recover $\overline{\langle D_i: 1\leq i\leq s\rangle}$ as the set of diagonal matrices in $P^{-1}\,\overline{\langle M_i: 1\leq i\leq s\rangle}\, P$. 
As in~\Cref{prop:semisimplematrix}, 
it follows that 
$\overline{\langle D_i: 1\leq i\leq s\rangle}$ is the zero locus of a system of polynomials of degree at most $b$.  

By \cref{prop:bound}, 
$\overline{\langle D_i: 1\leq i\leq s\rangle}=\{\Delta(g):g\in H_\Lambda\}$
for some lattice $\Lambda \subseteq \mathbb{Z}^d$ that is generated by vectors having supremum norm at most $b$ and $H_\Lambda$ is $s$-generated. 
Note that 
$\overline{\langle U_i: 1\leq i\leq s \rangle}$ is %
the set of upper unitriangular matrices in 
$P^{-1}\, \overline{\langle M_i: 1\leq i\leq s\rangle}\, P$.

The decision  procedure is as follows:
\begin{enumerate}
    \item \label{item 0}
    Guess a lattice $\Lambda \subseteq \mathbb{Z}^d$ whose generators have norm at most $b$ and such that $H_\Lambda$ is $s$-generated.
    \item Output "yes" %
    if there exists $P\in \mathrm{GL}_d(\QQbar)$ 
        such that 
        \begin{enumerate}
            \item \label{item 1} $G\coloneqq P^{-1}ZP$ is a commutative group of upper triangular matrices;
            \item \label{item 2} $\{A \in G : A \text{ diagonal}\} = \{\Delta(g): g\in H_\Lambda\}$; and
            \item \label{item 3} $\{ \log(A): A\in G, \; A \text{ unipotent} \}$
            is a linear variety of dimension at most $s$.
        \end{enumerate}
 \end{enumerate}%
 
As in~\Cref{prop:semisimplematrix}, guessing~$\Lambda$ that satisfies the condition of \cref{item 0} 
is via~\cref{prop:cyclic}.
We claim that the procedure outputs "yes" %
if and only if the input is a positive instance of the problem. 
Indeed, \cref{item 1} checks that $G$ is a commutative algebraic matrix group.  In this case, by~\Cref{fact:decomp}, both the set $G_s$ of semisimple matrices in $G$ and the set $G_u$ of unipotent matrices in $G$ form subgroups of $G$.
Next, \cref{item 2,item 3} respectively check that
$G_s$ and $G_u$ are $s$-generated (relying on \cref{prop:cyclic,prop:uni}). 
This in turn implies that $G$ is itself $s$-generated, as noted in~\cref{Theremark}.

The existence of $P$ 
satisfying \cref{item 1,item 2,item 3} reduces to checking the truth in $\QQbar$ of an $\exists^*\forall^*$-sentence in the language of  fields.  The existential quantifiers correspond to the possible choices of $P$, while the universal quantifiers range over entries of the group $G$ defined in \cref{item 1}.  For a fixed choice of $\Lambda$, the truth of such a formula can be decided in time $(mb)^{\mathrm{poly}(d)}$ by \cref{thm:FO}.
Given that the number of possible choices of the lattice $\Lambda$ is at most $(2b)^{d^2}$ the claimed complexity bound immediately follows. 
\end{proof}

\section{Orbit-Closure Determination}
\label{sec:orbitdet}
Recall the aforementioned {orbit-closure determination problem}.  The problem asks,  given   $s \in \N$ and a family
 of $m$ polynomials in $\QQ[\{x_i\}_{1\leq i\leq d}]$, each of total degree at most $b$,  to determine whether 
 their zero locus $Z\subseteq \QQbar^d$ is 
 the orbit closure of some point~$\bm{v}\in \QQbar^d$ under the action of an  $s$-generated algebraic matrix group, i.e.,
 determine whether there exists an algebraic group topologically generated by matrices
    $M_1, \ldots, M_s \in \mathrm{GL}_d(\QQbar)$ and vector 
    $\boldsymbol{v} \in \QQbar^d$ such that
    $Z$ is the 
    Zariski closure of the orbit
    $\{ M_1 ^{\ell_1}\cdots M_s^{\ell_s} \boldsymbol{v} : \ell_1, \ldots, \ell_s \in \mathbb{Z}\}$.

Key to the main result in this section (\cref{pro:loopGen}) is \cref{lemma:MatrixUlike}.
Briefly, the lemma describes a change of basis on a given commutative matrix group \(G\) and vector $\bm{v}$ such that  in this new basis the elements of the semisimple subgroup \(G_s \leq G\) are diagonal,  the elements of the unipotent subgroup \(G_u \leq G\) are upper unitriangular; and moreover \(\bm{v}\) is mapped to a binary vector with a prescribed structure.
In fact, the lemma permits us to focus on 
the orbit closure of the vector~$\bm{1}$ (that is, the all-one vector).

\begin{lemma}\label{lemma:MatrixUlike}
        Let $G = \overline{ \langle M_i: 1\leq i\leq s \rangle} \leq \GL_d(\alg)$ be a 
  commutative algebraic group and $\bm{v} \in \QQbar^d$.  There exist $ P \in \GL_d(\alg)$,  
 and $T \in \{0, 1\}^{d\times k}$ with the following properties: 
        \begin{enumerate}
         \item  the elements of $PGP^{-1}$ consist of block diagonal matrices, where each block is a scalar multiple of a unitriangular matrix;
            \item $PG_sP^{-1}$ consists of diagonal matrices and $PG_uP^{-1}$ consists of upper unitriangular matrices,
            \item  $T^\top T=\Id_k$,
            \item $ P \bm{v} = T \bm{1}$ is a $0$-$1$ vector with at most one $1$ per block.
       \end{enumerate}
    \end{lemma}
    \begin{proof}

The proof of \cref{lemma:MatrixUlike} relies on
\cref{thm:basis} and~\Cref{col:somethingcomuttaive} in  Section~\ref{ssec:basis}, immediately below.
These results show how to find $P$ satisfying Item (1) such that 
\(P\boldsymbol{v} = \bm{e}_{i_1} + \cdots + \bm{e}_{i_k}\) is a sum of standard unit vectors of $\QQbar^d$.     
     Defining $T$ to be the $d\times k$ matrix with columns $\bm{e}_{i_1}, \ldots, \bm{e}_{i_k}$,
    it follows that \(T^\top T = \mathrm{Id}_k\) and \(P\boldsymbol{v}= T\boldsymbol{1}\).

    \end{proof}

\subsection{A Convenient  Base Change for Commuting Matrices}
\label{ssec:basis}
This section contains technical details concerning the proof of~\cref{lemma:MatrixUlike}.

  \begin{theorem}

Let $A_1,\ldots,A_m$ be pairwise commuting matrices   in $\QQbar^{d \times d}$, and 
$\boldsymbol v \in \QQbar^d$ a non-zero vector.  
 There exists $P \in \mathrm{GL}_d(\QQbar)$ such that 
 \begin{itemize}
     \item $PA_iP^{-1}$ is block diagonal;
     \item each block is  a
scalar multiple of an upper unitriangular matrix; 
\item $P\boldsymbol v$
is  a 0-1 vector  with at most one 1
per block.
 \end{itemize}
\label{thm:basis}
    \end{theorem}
    \begin{proof}
By~\cite[Theorem 1]{Newman1967TwoCT},  there exists $P \in \mathrm{GL}_d(\QQbar)$ such that $PA_iP^{-1}$ is upper triangular for $i\in \{1,\ldots,m\}$. 
We have the following easy refinement of this well-known fact.

\begin{claim}
\label{claim:1}
There exists $P \in \mathrm{GL}_d(\QQbar)$ such that $PA_iP^{-1}$ is upper triangular for $i\in
\{1,\ldots,m\}$ and $P \boldsymbol v = \boldsymbol e_j$ for some $j \in \{1,\ldots,d\}$.
\end{claim}
\begin{proof}[Proof of \cref{claim:1}]
As mentioned above,  by~\cite[Theorem 1]{Newman1967TwoCT},  there exists $P \in \mathrm{GL}_d(\QQbar)$ such that $PA_iP^{-1}$ is upper triangular for $i\in \{1,\ldots,m\}$.
This  yields a basis
$\boldsymbol v_1,\ldots,\boldsymbol v_d$ of
$\QQbar^d$ such that 
$A_i \boldsymbol v_j  \in \mathrm{span}(\boldsymbol v_1,\ldots,\boldsymbol v_j)$ for all $i$ and $j$. 
Let $\ell$ be the minimum index such that
$\boldsymbol v \in \mathrm{span}(\boldsymbol v_1,\ldots,\boldsymbol
v_\ell)$.  Then
$\boldsymbol v = \sum_{i=0}^\ell \alpha_i \boldsymbol v_i$, where
$\alpha_\ell \neq 0$ by minimality of $\ell$.  Hence
$\boldsymbol v_\ell \in \mathrm{span}(\boldsymbol
v_1,\ldots,\boldsymbol v_{\ell-1},\boldsymbol v)$.
Consider a new basis in which $\boldsymbol v_\ell$ is replaced by 
$\boldsymbol v$. 
Note that $\mathrm{span}(\boldsymbol v_1,\ldots,\boldsymbol v_\ell) =
\mathrm{span}(\boldsymbol v_1,\ldots,\boldsymbol v_{\ell-1},\boldsymbol 
v)$.

Taking $Q$ to be the matrix with columns
$\boldsymbol v_1,\ldots,\boldsymbol v_{\ell-1},\boldsymbol
v,\boldsymbol v_{\ell+1},\ldots,\boldsymbol v_n$ we have that
$Q\boldsymbol e_\ell = \boldsymbol v$. Let $P \coloneqq Q^{-1}$ then we have  that $PA_iP^{-1}$ is upper
triangular for $i\in \{1,\ldots,m\}$. This concludes the proof of the claim. 
  \end{proof}

      By \cite[Theorem 2.4]{OuakninePPW19}
we can write $\QQbar^d = V_1 \oplus \cdots \oplus V_k$ such that
each subspace $V_i$ is invariant under $A_1,\ldots,A_m$ and
for all $1\leq i\leq m$ and $1\leq j \leq k$ the restriction of $A_i$
to $V_j$ has a single eigenvalue.
Let $\boldsymbol v= \boldsymbol v_1 +\cdots + \boldsymbol v_k$ be the
corresponding decomposition of $\boldsymbol v$.

Fix a subspace $V_j$.  We apply Claim~\ref{claim:1}
to the restrictions of $A_1,\ldots,A_m$ to $V_j$ and the vector
$\boldsymbol v_j$.  We obtain a basis for $V_j$ that includes $ \boldsymbol v_j$
and such that the restriction of $A_i$ to $V_j$ is upper triangular.
      \end{proof}

The next  corollary follows from~\Cref{thm:basis}.

 \begin{corollary}
 \label{col:somethingcomuttaive}
        Let $A_1,\ldots,A_m$ be matrices in $\mathrm{GL}_d(\QQbar)$ and let
        $P \in \mathrm{GL}_d(\QQbar)$
        be such that $PA_iP^{-1}$ is block diagonal where each block is the
scalar multiple of an upper unitriangular matrix. 
Given $M \in \langle A_1,\ldots,A_m\rangle$, $M$ is semi-simple if and
only if $PMP^{-1}$ is diagonal and $M$ is unipotent if and only if
$PMP^{-1}$ is upper unitriangular.
        \end{corollary}
        \begin{proof}
          Clearly $PMP^{-1}$ is block diagonal.  Thus 
 $M$ is semisimple  if and only if it is
blockwise semisimple and it is unipotent if and only if it is
blockwise unipotent.  Moreover, since each block is upper triangular
with constant diagonal, such a block is diagonalisable if and only if it is
already diagonal and is unipotent if and only if it is has all ones along the diagonal.
          \end{proof}

\subsection{Decision Procedures for the Orbit-Closure Determination}
\label{ssec:orbitclosure}
The main contributions of this section are the procedures for certain cases of the orbit-closure determination problem (\cref{prop:loopGensemi,pro:loopGen}).
The procedure in \cref{prop:loopGensemi} supposes %
that the generators of the matrix group are semisimple and  commutative.
The procedure in \Cref{pro:loopGen} lifts the requirement that the generators are semisimple.
We illustrate the procedures with a worked example (\cref{ex:loopGen}) at the close of this section.

\begin{proposition}
\label{prop:loopGensemi}
The orbit-closure determination problem for commutative semisimple matrices with  parameters \((s,d,m,b)\) can be decided in time 
   $(mb)^{\mathrm{poly}(d)}$, and in space bounded by~$(d \log{b})^{O(1)}$. 
\end{proposition}
\begin{proof}

 Recall that the input to the problem consists of~$m$ polynomials, each of total degree at most~$b$, together with a natural number $s$.  Let $Z$ be the subvariety  of $\QQbar^d$ defined by the input polynomials. 
   The task is to determine whether  $Z$ is an orbit-closure under a  group that is generated by at most $s$  commutative  semisimple matrices.

Suppose that the input is a positive instance of the problem, that is, 
 $Z=\overline{\langle M_i: 1\leq i\leq s\rangle  \cdot \boldsymbol{v}}$ for  
 semisimple commutative matrices $M_1, \ldots, M_s \in\mathrm{GL}_d(\QQbar)$ and a vector
$\boldsymbol{v}\in \QQbar^d$.  
By \cref{lemma:MatrixUlike},  there exist matrices  
$P\in \mathrm{GL}_d(\QQbar)$ and 
$T\in \{0,1\}^{d\times k}$
such that $D_i\coloneqq PM_iP^{-1}$ are  diagonal, and moreover 
$T^\top T = \mathrm{Id}_k$
and $P\boldsymbol{v}=T\boldsymbol{1}$ hold.

For all $i\in \{1,\ldots,s\}$, denote by $D'_i \in \mathrm{GL}_k(\QQbar)$  the diagonal matrix uniquely defined by the requirement that $D_iT=TD'_i$. 
Write  
\(G\coloneqq\left\{g\in\mathbb{G}_m^k:\Delta(g) \in \overline{\langle D'_i: 1\leq i \leq s \rangle} \right\}\).
Then we have
 \begin{equation*}
 PZ = \overline{\langle D_i: 1\leq i \leq s  \rangle \cdot T \boldsymbol{1}} 
  = \overline{T \langle D'_i: 1\leq i \leq s \rangle \cdot \boldsymbol{1}} 
 = \overline{\{ Tg : g \in G\}}.
 \end{equation*}
Note moreover that by the definition of $G$ we have  \[G=\left\{ g\in \mathbb{G}_m^k : T g   \in PZ\right\}\] and so $G$ is defined by polynomials of total degree at most $b$. 
It follows from \cref{prop:bound} that $G=H_\Lambda$ for some lattice $\Lambda\subseteq \mathbb{Z}^k$ that is generated by vectors whose entries have absolute value at most $b$.  

Conversely, suppose that 
 $PZ= \overline{\left\{ Tg  : g \in H_\Lambda \right\}}$
for some matrices $P \in \mathrm{GL}_d(\QQbar)$ and $T \in \{0,1\}^{d\times k}$ such that $T^\top T=\mathrm{Id}_k$, and
lattice $\Lambda$ as above.  Then \[Z=\overline{\langle M_i : 1\leq i\leq s \rangle \cdot \boldsymbol{v}},\] 
where $\boldsymbol{v}=P^{-1}T\boldsymbol{1}$  and 
$M_i\coloneqq P^{-1}D_iP$
with
$D_i\in\mathrm{GL}_d(\QQbar)$ being any diagonal matrices such that  $D_iT=T\Delta(g_i)$ for some topological generators $\{g_1, \ldots, g_s\}$ of $H_\Lambda$.

In summary, the decision procedure is as follows:
\begin{enumerate}
\item Guess $k \subseteq \{0,\ldots,d\}$ and $T \in \{0,1\}^{d\times k}$ such that $T^\top T=\mathrm{Id}_k$. %
\item Guess a lattice $\Lambda \subseteq \mathbb{Z}^{k}$ whose generators have norm at most $b$ and such that $H_\Lambda$ is $s$-generated (see~\Cref{prop:cyclic}).
\item Determine whether there exists $P\in \mathrm{GL}_d(\QQbar)$ such that 
$PZ = \overline{\left\{Tg : g \in H_\Lambda\right\}}$. 
\end{enumerate}

Step~3 can be reduced in polynomial time to checking the truth of a $\exists^*\forall^*\exists^*$-sentence in the theory of 
real closed fields.  The outermost existential quantifiers correspond to the choice of the matrix $P$.  Then 
the right-to-left inclusion in the equation
$PZ = \overline{\left\{Tg : g \in H_\Lambda\right\}}$
is expressed by the formula 
\[\forall g \in H_\Lambda \, \exists \boldsymbol z \in Z\,  (P\boldsymbol z = Tg),\]
while, by \cref{fact:EucZar},
the left-to-right inclusion 
is expressed by the formula
\[ \forall \boldsymbol z \in Z \, \forall \varepsilon >0\, \exists g \in H_\Lambda \,\left(\|P\boldsymbol z - Tg \|<\varepsilon\right).
\]
By \cref{thm:FO}, the truth of such a sentence can be decided in time $(mb)^{\mathrm{poly}(d)}$.
The claimed overall complexity bound now follows from the fact that there are at most $(2b)^{d^2+1}$ choices of the lattice 
$\Lambda$ and matrix $U$.
\end{proof}

We motivate the constructive subroutines in \cref{prop:loopGensemi} with loop synthesis.

\begin{example} \label{ex:loopGensemi}
 Let us determine whether there is a single-path linear loop with update matrix \(M\in\GL_2(\QQbar)\) and initial vector \(\bm{v}\in\QQbar^2\) such that the zero set  $Z\subseteq\QQbar^2$ of the ideal $I \coloneqq \langle 4 x^2 + y^2 + 4x y - x - y \rangle$ satisfies \(Z= \overline{\langle M \rangle \cdot \bm{v}}\).
In other words, we seek a matrix \(M\) and vector \(\bm{v}\) in the loop
\begin{algorithmic}
 \State \(\bm{x} \leftarrow \bm{v}\);
\While{$(*)$} %
 \State \(\bm{x} \leftarrow M \bm{x}\);
\EndWhile
\end{algorithmic}
such that following two Hoare triples are satisfied:
\begin{align*}
   & \{ \mathbf{true}\} \; \boldsymbol x \leftarrow \boldsymbol v \;  \; \{ 4 x^2 + y^2 + 4x y - x - y=0\} \\
   & \{ 4 x^2 + y^2 + 4x y - x - y=0 \} \; \bm{x} \leftarrow M \bm{x} \; \{ 4 x^2 + y^2 + 4x y - x - y=0\} 
\end{align*}
and such that $4 x^2 + y^2 + 4x y - x - y=0$ is moreover the strongest polynomial invariant satisfied by all reachable program state (i.e., all other invariants 
lie in the ideal $I$).

Suppose that in Steps 1 and 2 of the procedure in \cref{prop:loopGensemi} we guess $H_\Lambda \coloneqq \{ (x,y) \in \mathbb{G}_m^2 : x^2 - y =0 \}$ and the matrix $T \coloneqq \Id_2$.  For Step 3, we want to find all invertible matrices $P = \left( \begin{smallmatrix}
        a & b \\
        c&d 
    \end{smallmatrix} \right) \in \GL_2(\QQbar)$ such that $V(P^{-1} \cdot I) = \overline{\{Th : h \in H_\Lambda\}}$. This is equivalent to the requirement that the two polynomials
    \begin{equation*}
    4(ax + by)^2 + (cx + dy)^2 + 4(ax+ by)(cx + dy) - ax - by - cx - dy
    \end{equation*}
    and
   \(
    x^2 -y
   \)
    are multiples of one another. 
    Therefore the polynomials %
    defining $P^{-1}$ comprise the following ideal:
    \begin{multline*}
        J_{P^{-1}} \coloneqq \langle  4a^2 + c^2 + 4ac - b - d, 4b^2 + d^2 + 4bd,
         8 ab + 2cd+ 4 ad + 4bc , a+ c \rangle 
         \\
       = \langle  a+c, c^2 - b -d, 2 b c + c d , (2b+d)^2 \rangle.
    \end{multline*}
     One choice of $P^{-1}$ is $\left( \begin{smallmatrix}
        \phantom{-}1 & -1 \\
        -1 & \phantom{-}2
    \end{smallmatrix} \right) $.  
   Thus we realise \(Z\) as the orbit closure \(\overline{\langle M \rangle \cdot \bm{v}}\)
    where
    \begin{equation*} 
    M \coloneqq P^{-1} \begin{pmatrix}
        2 &0\\
        0&4
    \end{pmatrix} P = \begin{pmatrix}
        0&-2\\
        4& \phantom{-}6
    \end{pmatrix} \quad \text{and} \quad \boldsymbol{v} \coloneqq P^{-1}\begin{pmatrix}
        1\\
        1
    \end{pmatrix} = \begin{pmatrix}
        0\\
        1
    \end{pmatrix}.
\end{equation*}
\hfill $\blacktriangleleft$
\end{example}
\color{black}

The following theorem is our main contribution, which provides a decision procedure for the orbit-closure determination problem for  commutative groups. The generalization  of this result to 
 the case of general matrix groups appears to be  challenging.

\bigtheorem*

\begin{proof}
 Recall that the input to the problem consists of~$m$ polynomials, each of total degree at most~$b$, together with a natural number $s$.  Let $Z$ be the subvariety  of $\QQbar^d$ defined by the input polynomials.
 The task is to determine whether $Z$ is an orbit closure under an algebraic  group that is topologically generated by at most $s$  commutative   matrices.

Suppose that the input is a positive instance of the problem, that is, 
 $Z=\overline{\langle M_i: 1\leq i\leq s\rangle  \cdot \boldsymbol{v}}$ for  
  commutative matrices $M_1, \ldots, M_s \in\mathrm{GL}_d(\QQbar)$ and a vector
$\boldsymbol{v}\in \QQbar^d$. 
Using \cref{fact:decomp}, let
 $G_u $ be the subgroup of unipotent elements of $G$ and $G_s$ be the subgroup of semisimple elements of  $G$. 
  We have \(Z = \overline{G_u \cdot  G_s \cdot \bm{v}}\).

Applying  \Cref{lemma:MatrixUlike} to the group $G$, we obtain
$P\in \mathrm{GL}_d(\QQbar)$
and $T \in \{0,1\}^{d\times k}$ 
such that $PG_sP^{-1}$ is a group of diagonal matrices, 
$PG_uP^{-1}$ is a group of upper unitriangular matrices,
$T^\top T = \mathrm{Id}_k$,
and $P\boldsymbol{v}=T\boldsymbol{1}$.
In particular, 
we have $PG_sP^{-1} = \overline{\langle D_i : 1\leq i \leq s\rangle}$
for diagonal matrices $D_1,\ldots,D_s$ and 
 $PG_uP^{-1} = \overline{\langle U_i : 1\leq i \leq s\rangle}$ for upper unitriangular matrices 
 $U_1,\ldots,U_s$, such that $PM_iP^{-1}=D_iU_i$ for $1\leq i \leq s$.

For all $i\in \{1,\ldots,s\}$, denote by $D'_i \in \mathrm{GL}_k(\QQbar)$   the diagonal matrix uniquely defined by the requirement that $D_iT=TD'_i$.
Furthermore, write
\begin{equation*}
    G' \coloneqq \{g\in\mathbb{G}_m^k:\Delta(g) \in \overline{\langle D'_i: 1\leq i\leq s\rangle} \}.
\end{equation*}
Then  we have
\begin{align*}
PZ &= \overline{P \cdot G_u \cdot G_s \cdot \bm{v}}\\
&= \overline{ \langle U_i: 1 \leq i \leq s\rangle  \cdot   \, \langle D_i: 1 \leq i \leq s \rangle  \cdot \, P\boldsymbol{v}}\\
&= \overline{ \langle U_i: 1 \leq i \leq s\rangle  \cdot   \, \langle D_i: 1 \leq i \leq s \rangle  \cdot \, T\boldsymbol{1}}\\
&= \overline{\langle U_i: 1 \leq i \leq s \rangle  \cdot  \, T \langle D'_i: 1 \leq i \leq s  \rangle \cdot \boldsymbol{1}}\\
&= \overline{\langle U_i: 1 \leq i \leq s \rangle  \cdot  \, TG'}
\end{align*}
where the last equality follows from \Cref{prop:uni}.

To obtain a degree bound on~$G'$, 
we prove that  $G' = \{ g \in\mathbb{G}_m^k  : P^{-1} T g \in Z\}$.

\begin{claim}
\label{claim:intersection}
  $G' = \{ g \in\mathbb{G}_m^k  : P^{-1} T g \in Z\}$.
\end{claim}
\begin{proof}[Proof of \Cref{claim:intersection}] 
  The left-to-right inclusion is obvious. 
    Now let $g \in \mathbb{G}_m^k$  such that $P^{-1}Tg \in Z$.
    Since the Zariski closure of
    \begin{equation} \overline{\langle U_i: 1 \leq i \leq s \rangle  \cdot  \, TG'}
    \label{eq:SET}
    \end{equation}
    coincides with its Euclidean closure {(by \cref{fact:EucZar})}, there exists a sequence of vectors $(h_n)_{n\ge 0} $ in~\eqref{eq:SET}
        such that $\lim_{n \to \infty} h_n = Tg.$
By construction, the elements of~\eqref{eq:SET}
are given by blocks, therefore we can argue blockwise. Note that a block of $Tg$ has only one non-zero entry, let us consider one of them: 
\[
(0, \ldots, 0,g_i, 0 ,\ldots,0)^\top 
\]
where  $g_i$ corresponds to the $i$th entry in the block. 
The corresponding block of $h_n$ has the form
\[
(h_n^{(1)}, \ldots ,h_n^{(i)},0, \ldots, 0)^\top,
\]
and moreover the corresponding block of $P\bm{v}$ has $1$ in the $i$th entry and $0$ elsewhere. Therefore, we have
\begin{equation}\label{eq:limit}
\lim_{n \to \infty} \begin{pmatrix}
    h_n^{(i)} & * & *  & * & * \\
    & \ddots &  \ddots& * & *\\ 
     && \ddots  & * & * \\
     & & & &  h_n^{(i)}
\end{pmatrix} \begin{pmatrix}
     0 \\
     \vdots \\
     1 \\
     0 \\
     \vdots
     
\end{pmatrix} =  \begin{pmatrix}
     0 \\
     \vdots \\
     g_i \\
     0 \\
     \vdots
     
\end{pmatrix}
\end{equation}
where the matrix  
\[\begin{pmatrix}
    h_n^{(i)} & * & *  & * & * \\
    & \ddots &  \ddots& * & *\\ 
     && \ddots  & * & * \\
     & & & &  h_n^{(i)}
\end{pmatrix}\]
is a block of an element in $P G P^{-1}$ for every $n \in \N$, and its corresponding semisimple element has as  block  $h_n^{(i)} \Id $.  
From \eqref{eq:limit} we conclude that
\[
\lim_{n \to \infty } ( h_n^{(i)}\Id)\cdot  \begin{pmatrix}
     0 \\
     \vdots \\
     1 \\
     0 \\
     \vdots
     
\end{pmatrix}  =  \begin{pmatrix} 
     0 \\
     \vdots \\
     g_i \\
     0 \\
     \vdots
     
\end{pmatrix} 
\]
and moreover that the diagonal blocks $( h_n^{(i)}\Id)$ correspond to the blocks of diagonal matrices  that belong to $PGP^{-1}$. Thus, there exists a sequence of diagonal matrices $(D^{(n)})_{n \ge 0} \subseteq PGP^{-1}$ such that 
\[
\lim_{ n \to \infty} D^{(n)} T \bm{1} = Tg.
\]
Let  $\widetilde{D}^{(n)}$ be the diagonal matrix uniquely defined by  $D^{(n)}T=T \widetilde{D}^{(n)}$.  
Then we have $( \widetilde{D}^{(n)})_{n\ge 0} \subseteq \overline{\langle D'_i: 1\leq i\leq s\rangle}$. Hence we have, 
\[
\lim_{ n \to \infty } T\widetilde{D}^{(n)}\bm{1} = T\Delta(g)\bm{1},
\]
and since multiplication by $T$ defines an injective map we have 
\[
\lim_{ n \to \infty } \widetilde{D}^{(n)} = \Delta(g),
\]
which proves that $ g \in G'$.
This ends the proof of \cref{claim:intersection}.
\end{proof}

It follows that $G'$ is defined by polynomials of total degree at most $b$ and hence has the form $H_\Lambda$ for some lattice $\Lambda\subseteq \mathbb{Z}^k$ whose generators have norm at most $b$ and such that $H_\Lambda$ is $s$-generated (see~\Cref{prop:cyclic}).

Conversely, suppose that there exist  $k \in \{0,\ldots,d\}$ and $T \in \{0,1\}^{d\times k}$ such that $T^\top T=\mathrm{Id}_k$,
together with $P \in \GL_d(\QQbar)$,
commuting matrices $U_1,\ldots,U_s,D_1,\ldots,D_s\in \GL_d(\QQbar)$ with $U_i$ upper unitriangular and $D_i$ diagonal,
and an $s$-generated lattice $\Lambda\subseteq \mathbb Z^k$ such that \[T H_\Lambda = \overline{\langle D_1,\ldots,D_s \rangle \cdot T \boldsymbol 1}\]  satisfying
\begin{equation}
\label{eq:ZAfC}
     PZ = \overline{\langle U_i :1\leq i\leq s \rangle \cdot TH_\Lambda}.
\end{equation}
Then
\(
        Z=\overline{\langle M_i : 1 \leq i \leq s \rangle \cdot \boldsymbol{v}}\),
where $M_i\coloneqq P^{-1}U_i D_i P$ and $P\boldsymbol v = T\boldsymbol 1$.

In summary, the decision  procedure is as follows:
\begin{enumerate}
\item Guess $k \in \{0,\ldots,d\}$ and $T \in \{0,1\}^{d\times k}$ such that $T^\top T=\mathrm{Id}_k$.
\item Guess a lattice $\Lambda \subseteq \mathbb{Z}^{k}$ whose generators have norm at most $b$ and such that $H_\Lambda$ is $s$-generated, say by $g_1,\ldots,g_s \in H_\Lambda$ (see~\Cref{prop:cyclic}).
\item Return "yes" if there 
 exist $P\in \mathrm{GL}_d(\QQbar)$ and commuting matrices $U_1, \ldots, U_s, D_1,\ldots,D_s \in \mathrm{GL}_d(\QQbar)$ such that
 the $U_i$ are upper unitriangular, the $D_i$ are diagonal, $D_iT=T\Delta(g_i)$ holds for all $i\in\{1,\ldots,s\}$, and \eqref{eq:ZAfC} holds. 
\end{enumerate}

Step 3 can be reduced to checking the truth of an $\exists^*\forall^*\exists$-sentence 
over the theory of real closed fields.  The outer group of existential quantifiers range over the possible choices of the matrices
$P$ and $U_1,\ldots,U_s$.    The rest of the formula checks~\eqref{eq:ZAfC}.
The right-to-left inclusion is expressed by the formula
\[ \forall t_1 \cdots \forall t_s \, \forall g \in H_\Lambda \cdot  \exp(\textstyle{\sum_{i=1}^s} t_i\log U_i)Tg \in PZ.\]
By~\Cref{fact:EucZar}, the left-to-right inclusion is expressed by the formula
\[ \forall \boldsymbol z \in Z \, \forall \varepsilon>0 \exists t_1 \cdots \exists t_s \, \exists g \in H_\Lambda \cdot \|P\boldsymbol z-\exp(\textstyle{\sum_{i=1}^s} t_i\log U_i)Tg \|<\varepsilon. \]
By \cref{thm:FO}, the truth of such a sentence can be decided in time $(mb)^{\mathrm{poly}(d)}$.
Then the overall complexity bound follows from the fact that the number of choices of the lattice $\Lambda$ and matrix $T$ is at most $(2b)^{d^2+1}$.
\end{proof}

\cref{ex:loopGen}, below, applies the procedure in \cref{pro:loopGen} to the variety we first saw in \cref{ex:loopGenintro,ex:loopGenintro2}  in the Introduction.
The calculations involved in the preparation of \Cref{ex:loopGen,example-semi} were performed in \textsc{Macaulay2}~\cite{M2}.

\begin{example} \label{ex:loopGen}
    Let $Z\subseteq \QQbar^4$ be the zero set of the ideal \(I= \langle Q_1, Q_2 \rangle\), defined in \Cref{ex:loopGenintro2},  where
    \[  Q_1 = x_2^2 - x_1 - x_4 \quad \text{and} \quad Q_2 = -2 x_4 x_2 - 2x_3^2 - \frac{1}{5} x_2 x_3.\]
    
 Let us determine whether the zero set \(Z\) is equal to the orbit closure of a 1-generated algebraic matrix group. 
To answer this  affirmatively, we  construct a matrix \(M\in\GL_d(\QQbar)\) and vector \(\bm{v}\in\QQbar^d\) such 
    that \(Z\) is the orbit closure \(\overline{\langle M \rangle \cdot \bm{v}}\).

    Suppose that in Steps 1 and 2 of the procedure in \cref{pro:loopGen} we nondeterministically guess 
    \begin{equation*}
        H_\Lambda \coloneqq \{ (x,y) \in \mathbb{G}_m^2 : x^2 - y=~0\} \; \text{and} \;
        T \coloneqq \begin{pmatrix}
        0& 0\\
        0&0\\
        1&0\\
        0&1
    \end{pmatrix}. 
    \end{equation*}
      Let $g = ( 5, 25)$ be a topological generator of $H_\Lambda$, and let $$D = \begin{pmatrix}
          5& 0& 0& 0\\
          0&5&0&0\\
          0&0&5&0\\
          0&0&0&25
      \end{pmatrix},$$ which satisfies $DT= T \Delta(g)$.  
      
    For Step 3, we would like to find a set of invertible matrices \(P\subseteq \GL_4(\QQbar)\) (with $P^{-1} = (p_{ij})_{\{1 \leq i, j\leq 4\}}$) and a matrix 
   \begin{equation*}
       U \coloneqq \begin{pmatrix}
        1& \lambda&0&0\\
        0&1& \lambda &0\\
        0&0&1& 0\\
        0&0&0&1
    \end{pmatrix},
   \end{equation*}
   for some $\lambda \in \QQbar$, such that
    \[V(P^{-1}\cdot I)= \overline{\{\exp(t\log U)Th : h \in H_{\Lambda}, t\in \QQbar\}}.\] 
    Note that 
     \begin{equation*}
    \exp(t \log U) = \begin{pmatrix}
        1 & t \lambda & \frac{t( t-1)}{2} \lambda^2 &0\\
        0 & 1 & t \lambda & 0\\
        0& 0 & 1&0\\
        0&0&0&1
    \end{pmatrix}; 
    \quad \text{thus} \quad  \exp(t \log U) T \begin{pmatrix}
        x \\
        x^2
    \end{pmatrix} =\begin{pmatrix}
       \frac{t( t-1)}{2} \lambda^2 x \\
       t \lambda x\\
       x \\
       x^2
    \end{pmatrix}.
     \end{equation*}
    
    The ideal defining 
    \(\overline{\{\exp(t\log U)Th : h \in H_{\Lambda}, t\in \QQbar\}}\) is $H \coloneqq\langle x_3^2 - x_4, x_1 x_3 - \frac{1}{2} x_2^2 + \frac{\lambda}{2} x_2 x_3 \rangle$.
Consider the ideal 
    \begin{equation*}
 I_{P^{-1}} = \langle F_1(P^{-1} X), F_2(P^{-1}X) \rangle \subseteq \QQ[P^{-1} = ({p}_{ij}), \lambda, y]/ \langle (\det P^{-1})y -1 \rangle [x_1 , x_2, x_3, x_4].
 \end{equation*}
By applying  Algorithm 
\textsc{ContainmentIso}\footnote{The algorithm  \textsc{ContainmentIso} 
inputs two ideals $I_1$ and $I_2$ and outputs the locus of points $P^{-1}$ for which $ P^{-1} \cdot I_1 \subseteq  P^{-1}\cdot I_2$. 
This algorithm thereby solves a
generalisation of the ideal membership algorithm since it determines the
containment of an ideal into another after a change of variables. Clearly
\textsc{ContainmentIso} can also be applied to determine equality after a change of variables, since $P^{-1} \cdot I_1 =  P^{-1}\cdot I_2$ if and only if $P^{-1} \cdot I_1 \subseteq  P^{-1}\cdot I_2$ and $P^{-1} \cdot I_2 \subseteq  P^{-1}\cdot I_1$. See~\cite[Algorithm 2.9]{AlgsMateusz} for more details.}
and eliminating $\lambda$ we obtain the following ideal, defining the set of admissible choices of $P^{-1}$:
\begin{align*}
    J_{P^{-1}}\coloneqq \langle & p_{34},p_{31},p_{24},p_{22},p_{21},
    p_{13}+p_{43}, p_{12}+p_{42},\\ &
    p_{11}+p_{41},    p_{33}p_{44},p_{32}p_{44}, p_{23}p_{44}, p_{14}p_{44}+p_{44}^2, \\ &
    p_{23}p_{33}+10p_{33}^2+10p_{23}p_{43}, 
2p_{32}^2+p_{23}p_{41},p_{23}^2-p_{14}-p_{44}\rangle.
\end{align*}

One may choose, for example, 
    \begin{equation*}
        P^{-1} = \begin{pmatrix}
    \phantom{-}\frac{1}{2}& 0& 0& 1 \\
    \phantom{-}0& 0& 1& 0\\
    \phantom{-}0& \frac{1}{2}& 0& 0\\
    -\frac{1}{2}& 0& 0& 0
\end{pmatrix}   \quad \text{and} \quad \lambda = -\frac{1}{5}.
    \end{equation*}
Thus we realise  \(Z\) as the orbit closure \(\overline{\langle M \rangle \cdot \bm{v}}\), with \(M=P^{-1} U D P\) and \(\boldsymbol{v}= P^{-1} T\boldsymbol{1}\), that is
\begin{equation*}
M = \begin{pmatrix}
    25& \phantom{-}0& -1& 20\\
    0& \phantom{-}5& \phantom{-}0& 0 \\
    0& -\frac{1}{2}& \phantom{-}5& 0\\
    0& \phantom{-}0& \phantom{-}1& 5
\end{pmatrix} \quad \text{and} \quad \boldsymbol{v} = \begin{pmatrix}
    1\\
    1\\
    0\\
    0
\end{pmatrix}.
\end{equation*}
\hfill $\blacktriangleleft$
\end{example}

\section{Algorithms to Compute Generators}
\label{sec:generators}
In \cref{sec:group} we gave  algorithms
to determine whether a given variety is the Zariski closure of a commutative  matrix group.  These procedures
can also be used to find a set of generators of such a group, by  quantifier elimination in the theory of algebraically closed fields. 
In this section, we describe specialised algorithms that can be directly implemented using standard computer-algebra software.  The cases we consider compute 
 a generator when the input variety is the Zariski closure  of a cyclic group. 
 These algorithms  rely  on Gr\"{o}bner-basis techniques.
 The first algorithm finds a semisimple generator, if one exists, while the second algorithm finds a generator in the general case.

Let   $I \subseteq \QQbar[\bm{x}]$ be an ideal and
denote by~$\sqrt{I}$ the \emph{radical} of~$I$, defined as  
    \begin{equation*}
\sqrt{I}\coloneqq  \{f\in \QQbar[\bm{x}] \mid f^n\in I \text{ for some } n\in \N\}.
    \end{equation*}
By Hilbert's Nullstellensatz, the ideal of all polynomials that vanish on \(V(I) \in \QQbar^d\) is \(\sqrt{I}\).
The ideal \(I\) is \emph{primary} if \(fg\in I\) with \(f,g\in \QQbar[\bm{x}]\) implies that \(f\in I\) or \(g^n\in I\) for some \(n\in\N,\)
and it is \emph{prime} if it satisfies the stronger condition that 
$ fg \in I$  only if $f \in I$ or $g \in I$.
Recall that the radical of a primary ideal is necessarily prime. 

A polynomial ideal~$I$ can be written as the intersection of primary ideals, giving the so-called \emph{primary decomposition} of~$I$. It is known that there exists a unique irredundant primary decomposition $I=\bigcap_{i=1}^\ell Q_i$, that is, a finite set $\{ Q_1, \ldots, Q_\ell \}$ of primary ideals  such that
\begin{itemize}
    \item the prime ideals $\sqrt{Q_i}$ are all distinct; and
    \item $\bigcap_{i \neq j} Q_i \not\subseteq  Q_j$ holds for all $j\in \{1,\ldots, \ell\}$.
\end{itemize} 
The prime ideals, the $\sqrt{Q_i}$'s, are called the \emph{associated primes} of $I$. 
 An associated prime~$\sqrt{Q}$ of the ideal~$I$ is
 called \emph{minimal} if it does not contain any other
 associated primes of $I$.  

Both algorithms take as input a variety $Z$, given as the zero set of an ideal $I\subseteq \QQ[X]$, $X = \{x_{i,j}\}_{1\leq i,j\leq d}$.
We will assume \emph{a priori} that \(Z\) is a commutative subgroup of $\GL_d(\QQbar)$.
Verifying that \(Z\) is commutative entails first checking that $Z$ is
closed under matrix multiplication (which implies closure under matrix inversion), which amounts to showing 
that \[ F(XY) \in \sqrt{I(X)+I(Y)}\] for all 
polynomials~$F$ in  $I$,
where $X = \{x_{i,j}\}_{1\leq i,j\leq d}$, $Y = \{y_{i,j}\}_{1\leq i,j\leq d}$. 
Commutativity is captured by showing that \[XY-YX \in \sqrt{I(X)+I(Y)}\,.\] 

\subsection{Semisimple Generator}

In the following we describe a procedure that, given an ideal 
   $I \subseteq \mathbb{Q}[X]$, 
   $X = \{x_{i,j}\}_{1\leq i,j\leq d}$,
   determines whether there exists a
semisimple matrix 
   $M\in \mathrm{GL}_d(\QQbar)$
   such that $I$ is the vanishing ideal of the 
    group $\langle M\rangle$
    and which moreover outputs such an $M$ in case the answer is "yes".
We show that if such an $M$ exists then it can be chosen such that its eigenvalues lie in the number field $\QQ(\zeta_q)$, where~$q$ is the number of minimal associated primes of~$I$ and $\zeta_q$ is a primitive $q$th root of unity.

\begin{figure*}[t!]
    \caption{A procedure for the group determination problem of cyclic groups, specific  to   semisimple  generators.}
    \Description{Table for the procedure.}
    \centering
\setlength{\extrarowheight}{5pt}
    \begin{tabularx}{0.9\linewidth}{ p{0.1\linewidth} p{0.8\linewidth} }
   \rowcolor{Gray}
    \multicolumn{2}{c}{\textbf{Cyclic Groups: Semisimple Generator}} \\
    \textbf{Input:} & An ideal 
   $I \subseteq \mathbb{Q}[X]$ generated by $F_1,\ldots,F_k \in \QQ[X]$, 
   $X = \{x_{i,j}\}_{1\leq i,j\leq d}$, with $q$ minimal associated primes, such that $V(I)$ is a commutative linear algebraic group.
    \\  \rowcolor{Gray}
\textbf{Output:} &  Determine whether there exists a
semisimple matrix 
   $M\in \mathrm{GL}_d(\QQbar)$
   such that $V(I)=\overline{\langle M\rangle}$. If "yes", output such a matrix~$M$. 
     \\
       \textbf{Line 1:} & Define the  ideal $J_0 \subseteq \QQ[P,X,Q]$ as follows
 \[J_0\coloneqq   \langle \,  F_1(PXQ), \, \ldots, \, F_k(PXQ),\,  PQ-\Id_d, \,  \{x_{i,j}\}_{i\neq j} \, \rangle \,\]
where $P= \{p_{i,j}\}_{1 \leq i,j \leq d}$,  $X = \{x_{i,j}\}_{1\leq i,j\leq d}$, and 
 $Q=\{q_{i,j}\}_{1\leq i,j \leq d}$.  \\
  \rowcolor{Gray} \textbf{Line 2:} & 
  Write  $J\coloneqq   \sqrt{J_0 \cap \QQ[X]}$. Compute  the unique irredundant primary decomposition $J = \bigcap_{s \in S} \mathcal{P}_s$. 
\\
    \textbf{Line 3:} &  Check whether all  primary components $\mathcal{P}_s$ of~$J$ are  binomials using Gr\"obner basis computation; \textbf{return} "no" if this test fails.  \\
 \rowcolor{Gray} \textbf{Line 4:} & 
 Let $\mathcal{P}_0$ be one of the primary components of~$J$ such that $\Id_d \in V(\mathcal{P}_0)$. \\
  \textbf{Line 5:} & Following~\Cref{prop:cyclic},  construct a rational diagonal  matrix  $D$ for $\mathcal{P}_0$, such that all the entries of $\sqrt[q]{D}$ lie in $\QQ[\zeta_q]$.
   \\
  & Write $ D_q\coloneqq   \sqrt[q]{D}$. 
   \\ 
 \rowcolor{Gray} \textbf{Line 6:} & 
Check whether for all $i\in \{1,\ldots, q-1\}$, the ideal $D_q^i \cdot \mathcal{P}_0$ is  a primary component of $J$; \textbf{return} "no" if this test fails.\\
 \textbf{Line 7:} & Write $I_q\coloneqq   \bigcap_{ 1\leq i\leq q} \; D_q^i \cdot \mathcal{P}_0$. \\
  \rowcolor{Gray}  \textbf{Line 8:}  & 
 Check whether $J = \bigcap_{ \sigma \in S_d} \; M_\sigma I_q M_\sigma^{-1}$ where $M_\sigma$ is the permutation matrix corresponding to $\sigma \in S_d$;  \textbf{return} "no" if this test fails. \\
      \textbf{Line 9:} &
    Define the ideal 
    $J_1\coloneqq \langle F_1(QD_qP),\ldots,F_k(QD_qP), PQ-I \rangle \cap \QQ[P]$. \\
    
      & Pick $\widetilde{P}\in V(J_1)$.
      \\ 
      \rowcolor{Gray} \textbf{Line 10:} & Check whether $I= \widetilde{P}I_q\widetilde{P}^{-1}$; 
    \textbf{return} "no" if this test fails. 
      \\
       \textbf{Return:}  & "yes" together with the matrix~$\widetilde{P}D_q\widetilde{P}^{-1}$. 
      \\ \hline
     \end{tabularx}
    \label{fig:semisimplegroup}
   \end{figure*}

Let the input ideal~$I$ be generated by a finite collection of polynomials $F_1,\ldots,F_k \in \QQ[X]$ with $q$ minimal associated primes.  Write $Z\coloneqq  V(I)$ for the zero locus of $I$, assumed to be a commutative linear algebraic group.
The general procedure of the algorithm is depicted in~\Cref{fig:semisimplegroup}. 
 The ideal~$J_0$ defined in \textbf{Line 1} is an ideal of the ring $\QQ[P,X,Q]$, where the relations $\{x_{i,j}\}_{i\neq j}$ and $PQ-\Id_d$ ensure that every point $(\widetilde{P},\widetilde{X},\widetilde{Q}) \in V(J_0)$ comprises a diagonal matrix~$\widetilde{X}$ and an invertible matrix~$\widetilde{P}$ with $\widetilde{P}^{-1}=\widetilde{Q}$ satisfying $\widetilde{P} \widetilde{X} \widetilde{P}^{-1}\in Z$. 
The aim is to find a single such point $(\widetilde{P},\widetilde{X},\widetilde{Q}) \in V(J_0)$ satisfying  \[Z=\overline{\langle \widetilde{P} \widetilde{X} \widetilde{P}^{-1}\rangle}\,.\] 
Subsequently, the  ideal $J$ defined in \textbf{Line 2}  
 contains all diagonal conjugates of each matrix in~$Z$.
 
 In particular, for each matrix~$M\in Z$
 not only does the diagonal matrix~$D$ satisfying $M=\widetilde{P} D\widetilde{P}^{-1}$ lie in  $V(J)$,
 but also all the diagonal matrices \(D\) for which %
 $M_{\sigma} D M_{\sigma}^{-1}$ and the permutation $ \sigma  \in S_d$ lie in $V(J)$. 

Due to this fact, we cannot simply employ~\Cref{prop:cyclic} to construct a generator for~$J$.
Instead, in \textbf{Line 4}, we  isolate a primary component~$\mathcal{P}_0$ of $J$ containing~$\Id_d$. 
In the following line,  we apply~\Cref{prop:cyclic} to the  binomial ideal~$\mathcal{P}_0$ and construct a diagonal matrix~$D$ such that $V(\mathcal{P}_0)=\overline{\langle D\rangle}$.
Since $V(\mathcal{P}_0)$ is irreducible, the matrix~$D$  can be chosen to have rational entries, 
from which it follows %
that the entries of  $\sqrt[q]{D}$  lie in $\QQ[\zeta_q]$.

The assertion in~\textbf{Line 6} verifies whether the orbit of $D_q$ cycles between the 
irreducible  components of~$V(J)$; this ensures that $V(I_q)=\overline{\langle D_q\rangle}$ is included in  $V(J)$, where $I_q$ is defined in~\textbf{Line 7}. Next, our procedure checks whether $J$ is equal to the intersection of $M_{\sigma}I_q M_{\sigma}^{-1}$. The necessity of the latter test is due to the  above-mentioned fact that $V(J)$ 
 contains all diagonal conjugates of each matrix in $Z$; see~\Cref{example-semi}. The rest of the algorithm is straightforward.

\begin{example}
\label{example-semi}
Let \(F_1\coloneqq   2z+w\), \(F_2\coloneqq   2x-2y+3w\), and \(F_3\coloneqq   4y^2-4yw+w^2-4y+4w\).
Consider the following ideal as an input to the procedure in~\Cref{fig:semisimplegroup}:
\begin{equation*}
I \coloneqq   \langle F_1, F_2, F_3 \rangle \subseteq \QQ\left[\begin{pmatrix}
    x & z\\
    w & y
\end{pmatrix}\right].
\end{equation*}
The ideal $I$ is prime (meaning that $q=1$) and  $V(I)$ is a  commutative linear algebraic group. 
The output of our procedure shows that there exists $M$ such that $V(I)=\overline{\langle M\rangle}$, and such  that the eigenvalues of $M$ lie in $\QQ$.

Following the algorithms, the ideal $J$ defined in \textbf{Line 2} has two  primary components  
\[\mathcal{P} := \langle w,z,y^2-x \rangle \qquad \text{and} \qquad \mathcal{P}' := \langle  w,z,x^2-y \rangle.\] 
Since $I_2 \in V(\mathcal{P} \cap \mathcal{P}')$, we can pick any of these  ideals as $\mathcal{P}_0$ in \textbf{Line 4}.  Following~\Cref{prop:cyclic} in \textbf{Line 5}, we may construct 
diagonal matrices $D = \Delta(4,2)$ and $ D' = \Delta(2,4)$
such that 
\[V(\mathcal{P})=\overline{\langle \Delta(4,2)\rangle} \qquad \text{and} \qquad  V(\mathcal{P'})=\overline{\langle \Delta(2,4)\rangle}.\]
Clearly, matrices $D$ and $D'$ are conjugates under  permutation of diagonals, implying that the assertion in \textbf{Line 8} holds.
 (The above is an indication  (1) that permutations of matrices arising from one  choice of $\mathcal{P}_0$   under $M_{\sigma}$ are suitable for other possible choices of  $\mathcal{P}_0$, and (2) the necessity of the check in \textbf{Line 8}).
Following \textbf{Line 9} for $D_q= \Delta(2,4)$, defining the ideal 
\begin{equation*}
    J_1 := \langle F_1(Q D_q P), \,F_2(Q D_q P), \, F_3(Q D_q P) \rangle \cap \QQ[P],
\end{equation*}
we have that  
    \begin{equation*}
        J_1 = \langle p_3 - 2p_4, p_1 - p_2  \rangle \subseteq \QQ\left[\begin{pmatrix}
    p_1 & p_2\\
    p_3 & p_4
\end{pmatrix}\right].
    \end{equation*}
Subsequently, one choice for a  semisimple generator of $V(I)$  is the following matrix~$M$: 
\[
M := \begin{pmatrix}
    \phantom{-}6 & 2\\
    -4 & 0
\end{pmatrix} = \begin{pmatrix}
    \phantom{-}1 & \phantom{-}1\\
    -2& -1
\end{pmatrix} \begin{pmatrix}
    2 &0\\
    0 & 4
\end{pmatrix} \begin{pmatrix}
    \phantom{-}1 & \phantom{-}1\\
    -2  &-1
\end{pmatrix}^{-1}.
\]
\hfill $\blacktriangleleft$
\end{example}

\subsection{General  Generator}
We employ the algorithm from the previous subsection to provide a procedure that, given an 
   algebraic set $Z \subseteq \mathrm{GL}_d(\QQbar)$ determines whether
   there exists a matrix 
   $M\in \mathrm{GL}_d(\QQbar)$
   such that $Z=\overline{\langle M\rangle}$
    and which moreover outputs such an $M$ in the affirmative case.

Let the input ideal~$I$ be generated by a finite  collection of polynomials $F_1,\ldots,F_k \in \QQ[X]$ with $q$ minimal associated primes.  Write $Z\coloneqq  V(I)$ for the zero locus of $I$ that is, by assumption, a commutative linear algebraic group.
Our algorithm first calls a (modified variant of) the procedure in~\Cref{fig:semisimplegroup}, with the input ideal~$I$,
to check whether the subgroup~$G_s$ of all semisimple matrices in $Z$
is one-generated. 
The modification is as follows:
(1) the assertion  in \textbf{Line 10} is omitted (as this assertion requires that $Z$ is generated with a single semisimple matrix), and (2) the algorithm outputs  $D_q$ and the ideal $J_1$ defining the locus point of a suitable~$P$. 

After the  above subprocedure returns, our algorithm proceeds by verifying that  the  subgroup~$G_u$ of all unipotent matrices is one-generated. For this purpose, it checks  
  \begin{itemize}
  \item whether $V(I + \langle   (X-\Id_d)^n  \rangle  )$ is a commutative linear algebraic group; and 
      \item  whether $V(I + \langle   (X-\Id_d)^n  \rangle  )$ is  one-dimensional.
  \end{itemize} 

The algorithm returns~"no" if either of the subgroups $G_s$ or $G_u$
is not one-generated. Otherwise, the procedure defines
the ideal $H\subseteq \QQ[P,X,Q]$ by 
\[ H\coloneqq   \langle \,  F_1(PXQ), \, \ldots, \, F_k(PXQ),\,  PQ-\Id_d, \,  \{x_{i,j}\}_{j\neq i,i+1} \, \rangle\]
 where \[P= (p_{i,j})_{1 \leq i,j \leq d}, \qquad X = \{x_{i,j}\}_{1\leq i,j\leq d} \qquad \text{and } \qquad 
 Q=(q_{i,j})_{1\leq i,j \leq d}\,.\]
It returns "yes"  together with the matrix~$\widetilde{P}D_q\widetilde{X}\widetilde{P}^{-1}$ where $(\widetilde{P},\widetilde{X}) \in J+ H$.

\begin{acks}
Mahsa Shirmohammadi and Rida Ait El Manssour were supported by the
ANR grant VeSyAM (ANR-22-CE48-0005).
James Worrell was supported by EPSRC Fellowship EP/X033813/1. \newline
\includegraphics[height=1em]{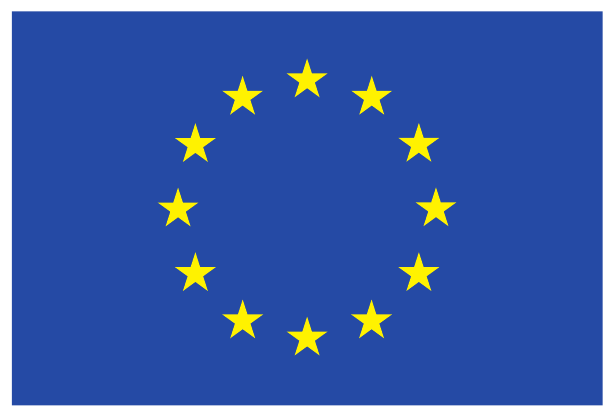}  This paper is part of a project that has received funding from the
European Research Council (ERC) under the European Union's Horizon 2020
research and innovation program (grant agreement No.~10103444).
Anton Varonka gratefully acknowledges the support of the ERC consolidator
grant ARTIST 101002685.

\end{acks}

\printbibliography%

\newpage

\appendix %
\section{Extended Background}
Here we give an extended preliminary section, which collects together background material and definitions relevant to the paper.
\label{app:preliminaries}
\subsection{Ideals and Varieties}

Let \(\QQbar\) denote the field of algebraic numbers and write $\QQbar[x_1,\ldots, x_d]$ for the
ring of polynomials with coefficients in $\QQbar$ over the variables~$x_1,\ldots,x_d$.  
A polynomial ideal~$I$ is an additive subgroup of~$\QQbar[x_1,\ldots, x_d]$
that is closed under multiplication in $\QQbar[x_1,\ldots, x_d]$.  
Given a finite collection of polynomials \(S\subseteq \QQbar[x_1,\ldots, x_d]\), we denote by \(\langle S\rangle\) the ideal generated by \(S\).

We briefly describe useful algebraic operations for ideals.
The \emph{sum} of \(I\) and \(J\), which we denote by \(I+J\), is the ideal 
    \begin{equation*}
        I + J \coloneqq \{f+g : f\in I \text{ and } g\in J\}.
    \end{equation*}
The sum \(I+J\) is the smallest ideal containing both \(I\) and \(J\) and, in addition, if \(I = \langle f_1, \ldots, f_r\rangle\) and \(J = \langle g_1, \ldots, g_s\rangle\), then \(I+J = \langle f_1,\ldots, f_r, g_1,\ldots, g_s\rangle\).
The \emph{product} of \(I\) and \(J\), which we denote by \(IJ\), is the ideal given by
    \begin{equation*}
        I J = \left\{ \textstyle \sum_{\ell=1}^r f_\ell g_\ell : f_1, \ldots, f_r\in I, g_1,\ldots, g_r\in J, r\in\N \right\}.
    \end{equation*}
The \emph{intersection} \(I \cap J\) is the set of polynomials common to both \(I\) and \(J\).
Given an ideal $I\subseteq \QQbar[X]$, where $X=\{x_{i,j}\}_{1\leq i,j\leq d}$, and matrix $M \in \QQbar^{d\times d}$, we write 
$M\cdot I$ for the ideal $\{f(MX) \in \QQbar[X] : f \in I\}$.

An \emph{algebraic set} (or \emph{variety}) is the set of common zeroes of a finite collection of polynomials.
By Hilbert's basis theorem every polynomial ideal \(I \subseteq \QQbar[x_1,\ldots, x_d]\) is finitely generated.
Thus the set 
    \begin{equation*}
        V(I) \coloneqq   \{\bm{a}\in\QQbar^d : f(\bm{a})=0 \text{ for all } f\in I\} 
    \end{equation*}
is a variety.
As an aside, the varieties \(V(I+J)\) and \(V(IJ)\) are easily shown to be equal to \(V(I) \cap V(J)\) and \(V(I)\cup V(J)\), respectively.
Taking, as above, \(M\in\QQbar^{d\times d}\), we have 
    \begin{equation*}
        V(M\cdot I)=\{ A \in \QQbar^{d\times d}: MA \in V(I) \}.
    \end{equation*}

Let   $I \subseteq \QQbar[x_1,\ldots, x_d]$ be an ideal and
denote by~$\sqrt{I}$ the \emph{radical} of~$I$, defined as  
    \begin{equation*}
\sqrt{I}\coloneqq  \{f\in \QQbar[x_1,\ldots,x_d ] \mid f^n\in I \text{ for some } n\in \N\}.
    \end{equation*}
By Hilbert's Nullstellensatz, the ideal of all polynomials that vanish on \(V(I) \in \QQbar^d\) is \(\sqrt{I}\).

A polynomial ideal \(I\) is \emph{principal} if there exists a polynomial \(f\) for which \(\langle f \rangle =I\).
The ideal $I$ is  \emph{primary}  if for all $f, g \in \QQbar[x_1, \ldots, x_d]$, if $ fg \in I$ then $f \in I$ or $g^n \in I$ for some $n \in \N$, and it is \emph{prime} if it satisfies the stronger condition that 
$ fg \in I$  only if $f \in I$ or $g \in I$.
Recall that the radical of a primary ideal is necessarily prime. 

A polynomial ideal~$I$ can be written as the intersection of primary ideals, giving the so-called \emph{primary decomposition} of~$I$. It is known that there exists a unique irredundant primary decomposition $I=\bigcap_{i=1}^\ell Q_i$, that is, a finite set $\{ Q_1, \ldots, Q_\ell \}$ of primary ideals  such that
(1) the prime ideals $\sqrt{Q_i}$ are all distinct; and   (2)  $\bigcap_{i \neq j} Q_i \not\subseteq  Q_j$ holds for all $j\in \{1,\ldots, \ell\}$.
The prime ideals, the $\sqrt{Q_i}$'s, are called the \emph{associated primes} of $I$. 
 An associated prime~$\sqrt{Q}$ of the ideal~$I$ is
 called \emph{minimal} if it does not contain any other
 associated primes of $I$.

Let \(k\) be a field.
Consider a system of \(m\) equations in \(d\) variables \(x_1,\ldots, x_d\) with coefficients in \(k\):
    \begin{align*}
        a_{11} x_1 + a_{12} x_2 + \cdots + a_{1d} x_d &= b_1 \\ 
        a_{21} x_1 + a_{22} x_2 + \cdots + a_{2d} x_d &= b_2 \\
         \shortvdotswithin{=} 
         a_{m1}x_1 + a_{m2} x_2 +\cdots + a_{md} x_d &= b_m.
    \end{align*}
We call the solution set for such a system a \emph{linear variety}.
The \emph{dimension} of a non-empty linear variety  is the number of independent equations in its definition.

\subsection{The Zariski Topology}

The \emph{Zariski topology} on 
 \(\QQbar^d\) has as its closed sets 
 the varieties in \(\QQbar^d\).
Given a set~\(E\subseteq \QQbar^d\), we denote by \(\overline{E}\) the closure of \(E\) in the Zariski topology, i.e., 
the smallest algebraic set that contains \(E\). 
A closed set $A \subseteq \QQbar^d$ is \emph{irreducible} if it cannot be written as the union of two closed proper subsets.  A maximal irreducible closed subset of $A$ is called 
an \emph{irreducible component} of $A$.
By Hilbert's basis theorem every closed set $A$ can be written as a finite union of its irreducible components.

Given a closed set~$E$, denote by $\dim E$ the dimension of the variety~$E$, that is, 
the maximal length of a strictly decreasing  chain of nonempty irreducible subvarieties of $E$. We define the dimension of an arbitrary set to be the dimension of
its closure. 

\subsection{Linear Algebraic Groups}
The general linear group $\GL_d(\mathbb{F})$ is the group of all $d\times d$ invertible matrices with entries in a given field~$\mathbb{F}$.
The orbit closure of $\bm{v} \in \mathbb{F}^d$, denoted by $\overline{G \cdot \bm{v}}$, is
the closure of $G\cdot \bm{v}$ in the Zariski topology.
Recall that 
the closed sets of the Zariski topology are algebraic sets--that is, sets of common zeros of a finite collection of polynomials.

Recall
that a matrix \(M\in \QQbar^{d\times d}\) is  
\emph{nilpotent} if \(M^n=0\) for some $n\in \NN$. 
It is
\emph{unipotent} if \(M -\textrm{Id}_d\) is nilpotent, and 
\emph{semisimple} if it is diagonalisable over~\(\QQbar\).
A matrix \(M\in \QQbar^{d\times d}\) is called 
 \emph{upper triangular}  if all entries below the main diagonal are zero.
We use the term \emph{upper unitriangular} to refer to an upper triangular matrix whose entries along the main diagonal are all ones.

Write  $\GL_d(\QQbar)$
for the group of  $d\times d$ invertible matrices with entries in $\QQbar$.
We identify $\mathrm{GL}_d(\QQbar)$ with the variety
\(\bigl\{(M,y)\in \QQbar^{d\times d}\times \QQbar: \det(M)\cdot y=1\bigr\}\).
Under this identification, matrix multiplication is a polynomial map $\GL_d(\QQbar)\times \GL_d(\QQbar) \rightarrow \GL_d(\QQbar)$, and, by Cramer's rule, matrix inversion is also a polynomial map $\GL_d(\QQbar) \rightarrow \GL_d(\QQbar)$. 
A \emph{linear algebraic group} $G$ is a Zariski-closed subgroup of $\GL_d(\QQbar)$.
Given a linear algebraic group $G$, it is well-known that it has a unique irreducible component that contains the identity matrix, which we denote by~$G^{\circ}$. 
The cosets of $G^{\circ}$ are the irreducible components of~$G$. 

We say that $G$ is \emph{topologically generated by}
$S\subseteq \GL_d(\QQbar)$ if $G$ is the smallest Zariski closed subgroup of 
$\GL_d(\QQbar)$ that contains $S$, that is, $G=\overline{\langle S\rangle}$.
If $G$ is topologically generated by a set with $s$ elements then we say that is
$G$ is \emph{$s$-generated}. 
Denote by  $G_s$  the subset of semisimple  matrices in $G$, and by $G_u$ the subset of unipotent matrices.
If $G$ is a commutative algebraic group then  $G_s$ and $G_u$ form  algebraic subgroups; moreover we have $G=G_u \cdot G_s$.

The $d$-dimensional multiplicative group 
over $\QQbar$ is defined by
 \begin{gather*}
 \mathbb{G}_m^d = \mathbb{G}_m^d(\QQbar) \coloneqq   \left\{\boldsymbol{a} \in \QQbar^{d} : a_1\cdots a_d
   \neq 0 \right\}.
 \end{gather*}
   Here the subscript $m$ stands for \emph{multiplicative}.
Evidently this is a commutative group with respect to pointwise multiplication.   

We identify $\mathbb{G}_m^d$ with the subgroup of diagonal matrices in $\GL_d(\QQbar)$ via the map $\Delta$
that sends $(a_1,\ldots,a_d) \in \mathbb{G}_m^d$
to the diagonal matrix $\Delta(a_1,\ldots,a_d) \in \GL_d(\QQbar)$.
   
 Given a subgroup $\Lambda \subseteq \mathbb{Z}^d$, define
 \[ H_\Lambda \coloneqq   \{ \boldsymbol{a} \in \mathbb{G}_m^d : \forall
   \boldsymbol{v}\in \Lambda\, (a_1^{v_1}\cdots a_d^{v_d}=1) \}.  \]
 The map $\Lambda \mapsto H_{\Lambda}$ is an isomorphism
 between subgroups of $\mathbb{Z}^d$ and algebraic subgroups
 of $\mathbb{G}_m^d$.  
 This implies that $\mathbb{G}_m^d$ is topologically generated by any $d$-tuple $(g_1,\ldots,g_d)$ of multiplicatively
independent elements of $\QQbar$.
It also follows 
 that the vanishing ideal 
$I\subseteq \QQbar[x_1,\ldots,x_d]$ of an algebraic subgroup of 
$\mathbb{G}_m^d$ is a so-called \emph{pure binomial ideal};
that is, an ideal generated by polynomials of the form $x_1^{\alpha_1}\cdots x_d^{\alpha_d}-x_1^{\beta_1}\cdots x_d^{\beta_d}$, where $\alpha_1,\ldots,\alpha_d$ and 
$\beta_1,\ldots,\beta_d$ are non-negative integers.  
A mere \emph{binomial ideal} is one that is generated by polynomials of the form
$x_1^{\alpha_1}\cdots x_d^{\alpha_d}-\lambda x_1^{\beta_1}\cdots x_d^{\beta_d}$, where $\lambda \in \QQbar$.

For a $d\times d$ unipotent matrix $A$ and  nilpotent matrix $B$, define 
\begin{equation*} \log(A) \coloneqq  \sum_{k=1}^{d-1} (-1)^{k+1}\frac{(A-I)^k}{k} 
\quad \text{and} \quad
\exp(B) \coloneqq   \sum_{k=0}^{d-1} \frac{B^k}{k!}.
\end{equation*}
Let \(G\subseteq \GL_d(\QQbar)\) be a commutative subgroup of unipotent matrices.
Recall that \(L\coloneqq   \{\log(A) : A\in G\}\) is a linear subspace of 
$\QQbar^{d^2}$
consisting of nilpotent matrices \cite[Chapter II, Section 7.3]{borel1991}.

\subsection{Lattices}
The \emph{rank} of an abelian group \(\Lambda\) is the size of a maximal linearly independent subset~\cite{Lang2002}. 
A subgroup \(\Lambda\subseteq \Z^d\)   is called a \emph{lattice} (and has rank at most~$d$).
The \emph{torsion subgroup} of~$\ZZ^d / \Lambda$ is the subgroup of $\ZZ^d / \Lambda$ consisting of all elements of finite order.

\subsection{Modules over Principal Ideal Domains}

Let \(R\) be a commutative ring.  A \emph{module} \(M\) is an additive abelian group together with an operation \(R\times M \to M\) such that for all \(r,s\in R\) and \(x,y \in M\) we have
    \begin{equation*}
        (r+s)x = rx + sx \quad \text{and} \quad r(x+y) = rx + ry.
    \end{equation*}
By definition of an operation, we have \(1x=x\) for all \(x\in M\).
A module over \(\Z\) is abelian group and, vice versa, an abelian group is a module over \(\Z\).

Let \(M\) be a module over a ring \(R\) and let \(S\subseteq M\).
The set \(S\) is a \emph{basis} of \(M\) if \(S\) is non-empty, \(S\) generates \(M\), and the elements of \(S\) are linearly independent.
A \emph{free module} is a module which admits a basis, or the zero module.
If \(S\) is a basis of a non-zero module \(M\), then every element of \(M\) can be uniquely expressed as a linear combination of elements of \(S\).

An \emph{integral domain} is a non-zero commutative ring where the product of two non-zero elements is itself non-zero.
A \emph{principal ideal domain} is an integral domain wherein every ideal is principal.  Examples of principal ideal domains include \(\Z\) and the rings \(\mathbb{F}[x]\), the univariate polynomials with coefficients in the field \(\mathbb{F}\).

 Broadly speaking, a finitely generated module over a principal ideal domain is given by a direct sum of cyclic modules (i.e., modules with a single generator).
 One formulation of this structural result is the Elementary Divisors Theorem, see, for example,~\cite{Lang2002}.
\begin{theorem}[Elementary Divisors Theorem]
    Let \(F\) be a free module over a principal ideal domain \(R\) and \(M\) a non-trivial finitely generated submodule.
    Then there exists a basis of \(F\), elements \(e_1, \ldots, e_m\) in this basis, and non-zero elements \(a_1,\ldots, a_m\in R\) such that
    \begin{enumerate}
        \item \(a_1e_1,\ldots, a_m e_m\) form a basis of \(M\) over \(R\), and
        \item \(a_i \mid a_{i+1}\) for \(i=1,\ldots, m-1\).
    \end{enumerate}
\end{theorem}

\section{Proofs Omitted from \texorpdfstring{\cref{sec:preliminaries}}{Section 2}}
\label{app:proofs}

Recall that the family of \emph{constructible sets} is the smallest class that contains the algebraic sets and is also closed under boolean operations~\cite[Chapter 1]{basu2006algorithms} and that a \emph{morphism} is a map that is given locally by polynomials.
\factEucZar*
\cref{fact:EucZar} follows from the observation that the orbit \(G\cdot \bm{v}\) is the image of a variety under a morphism and is thus a constructible set by Chevalley's Theorem (see~\cite[Chapter AG, Corollary 10.2]{borel1991}).
In our setting, the Zariski and Euclidean closures of a constructible set coincide~\cite[Chapter 1, \S10, Corollary 1]{Mumford1999}.

\propgen*
\begin{proof}
    Let $U$ be the group $ \overline{ \langle G_u \rangle}$. 
    It is shown in~\cite[Proof of Lemma 6]{NPSHW2021} that $U$ is a normal subgroup of~$G$, and 
   that it is topologically generated by \(\dim U\) elements.

      Recall that  the quotient of a linear algebraic group  by a normal
subgroup is itself isomorphic to a  linear algebraic group (possibly in higher dimension~\cite[Section 11.5]{HumphreysLAG}). 
By construction, the quotient $G/U$ is a linear algebraic group that consists only of semisimple elements. 
    Therefore, $(G/U)^{\circ} = G^{\circ}/U$ is a torus and by \cite[Proposition 14]{galuppi2021toric} it is $1$-generated. 
    
Let $h \in G^{\circ}$ be such that $G^{\circ}/U = \overline{\langle h \rangle}$, and define $H := \overline{ \langle h, U \rangle} \subseteq G^{\circ}$. Then $G^\circ /U = H / U$ and hence  $\dim H= \dim G^\circ$. Since $G^\circ$ is irreducible we have $G^\circ = H$; hence $G^{\circ}$ is topologically generated by at most $\dim U + 1$ elements. 
        
        In order to topologically generate $G$, it is sufficient to take the topological generators of $G^\circ$ and one element from every other irreducible component of $G$. Hence, $G$ is topologically generated by  $\dim U + | G/G^\circ|$ elements.
\end{proof}

\propcyclic*
\begin{proof}
Let $\Lambda \subseteq \ZZ^d$ have rank $r$ and elementary divisors $d_1,\ldots,d_r$, where $d_i \mid d_{i+1}$ for all $i\in \{1,\ldots,r-1\}$. 
Write $s_0$ for the number of non-unit elementary divisors.
Hence, 
there is a basis $\boldsymbol{u}_1,\ldots,\boldsymbol{u}_d$ of $\mathbb{Z}^d$ such that $\Lambda$ is generated by the vectors
 $d_1 \boldsymbol{u}_1 , \ldots, d_r \boldsymbol{u}_r$.
Then the torsion subgroup of 
$\mathbb{Z}^d/\Lambda$ is 
$\mathbb{Z}/d_1\mathbb{Z}\times \cdots 
\times \mathbb{Z}/d_r\mathbb{Z}$, which is $s_0$-generated.  This shows Item~1.

The map
 $\varphi \colon \mathbb{G}_m^d \rightarrow \mathbb{G}_m^d$, defined by
 $\varphi(\boldsymbol{a}) = (\boldsymbol{a}^{\boldsymbol{u}_1},\ldots,\boldsymbol{a}^{\boldsymbol{u}_d})$
is a Zariski-continuous group automorphism of $\mathbb{G}_m^d$ that maps $H_\Lambda$
to the group 
\( G\coloneqq  \Omega_{d_1}\times \cdots \times \Omega_{d_r} \times \mathbb{G}_m^{d-r}\), 
where $\Omega_k$ denotes the group of all $k$-th roots of unity for $k$ a positive integer.  Clearly $H_\Lambda$ is $s$-generated if and only if $G$ is $s$-generated.

Write $F\coloneqq  \Omega_{d_1}\times \cdots \times \Omega_{d_r}$, so  $G=F\times \mathbb{G}_m^{d-r}$ and $F$ is $s_0$-generated.
Hence if $r=d$ then $G$ is also $s_0$-generated, while if $r<d$ then $G$ is $s_0$-generated unless $s_0=0$, in which case $G$ is 1-generated.  
\end{proof}

\end{document}